\documentclass[journal]{IEEEtran}
\usepackage{amsthm,amssymb,amsmath,rangecite}
\usepackage{commath,amsmath,graphicx,epstopdf,amsthm,amssymb,float,cite,color,array}
\usepackage[ruled,vlined]{algorithm2e}
\newtheorem{Theorem}{Theorem}

\newtheorem{proposition}[Theorem]{Proposition}
\newtheorem{definition}{Definition}

\newtheorem{lemma}[Theorem]{Lemma}
\newcommand{\ud}{\,\mathrm{d}}
\usepackage{lscape}
\usepackage{cite}
\usepackage{url}
\usepackage{hyperref}

\newcommand{\avec}{{\bf{a}}}
\newcommand{\bvec}{{\bf{b}}}
\newcommand{\cvec}{{\bf{c}}}
\newcommand{\dvec}{{\bf{d}}}

\newcommand{\fvec}{{\bf{f}}}

\newcommand{\uvec}{{\bf{u}}}

\newcommand{\xvec}{{\bf{x}}}

\newcommand{\nvec}{{\bf{n}}}

\newcommand{\gvec}{{\bf{g}}}

\newcommand{\hvec}{{\bf{h}}}

\newcommand{\etavec}{{\bf{\eta}}}

\newcommand{\zerovec}{{\bf{0}}}

\newcommand{\Gammamat}{{\bf{\Gamma}}}

\newcommand{\Amat}{{\bf{A}}}
\newcommand{\Bmat}{{\bf{B}}}
\newcommand{\Cmat}{{\bf{C}}}
\newcommand{\Dmat}{{\bf{D}}}

\newcommand{\Fmat}{{\bf{F}}}

\newcommand{\Hmat}{{\bf{H}}}
\newcommand{\Jmat}{{\bf{J}}}

\newcommand{\Imat}{{\bf{I}}}

\newcommand{\Pmat}{{\bf{P}}}

\newcommand{\Smat}{{\bf{S}}}
\newcommand{\Tmat}{{\bf{T}}}

\newcommand{\Umat}{{\bf{U}}}
\newcommand{\Vmat}{{\bf{V}}}
\newcommand{\Wmat}{{\bf{W}}}

\newcommand{\define}{\stackrel{\triangle}{=}}

\def\psivec{{\mbox{\boldmath $\psi$}}}

\def\tauvec{{\mbox{\boldmath $\tau$}}}

\def\upsilonvec{{\mbox{\boldmath $\upsilon$}}}

\def\xivec{{\mbox{\boldmath $\xi$}}}

\def\gammavec{{\mbox{\boldmath $\gamma$}}}

\def\etavec{{\mbox{\boldmath $\eta$}}}

\def\thetavec{{\mbox{\boldmath $\theta$}}}

\def\lambdavec{{\mbox{\boldmath $\lambda$}}}

\def\deltavec{{\mbox{\boldmath $\delta$}}}

\def\thetavecsmall{{\mbox{\boldmath {\scriptsize $\theta$}}}}

\def\etavecsmall{{\mbox{\boldmath {\scriptsize $\eta$}}}}

\newcommand{\be}{\begin{equation}}
\newcommand{\ee}{\end{equation}}
\newcommand{\beqna}{\begin{eqnarray}}
\newcommand{\eeqna}{\end{eqnarray}}

\begin{document}
\title{Cram$\acute{\text{e}}$r-Rao Bound for Constrained Parameter Estimation Using Lehmann-Unbiasedness}
\author{Eyal~Nitzan,~\IEEEmembership{Student~Member,~IEEE,}
Tirza~Routtenberg,~\IEEEmembership{Senior~Member,~IEEE,}
and~Joseph~Tabrikian,~\IEEEmembership{Senior~Member,~IEEE}
\thanks{This research was partially supported by THE ISRAEL SCIENCE FOUNDATION (grant No. 1160/15 and grant No. 1173/16).}
\thanks{{\footnotesize{E. Nitzan, T. Routtenberg, and J. Tabrikian are with the Department of Electrical and Computer Engineering Ben-Gurion University of the Negev Beer-Sheva 84105, Israel, e-mail: eyalni@ee.bgu.ac.il, tirzar@bgu.ac.il, joseph@bgu.ac.il.}}}
}

\maketitle
\nopagebreak

\begin{abstract}
The constrained Cram$\acute{\text{e}}$r-Rao bound (CCRB) is a lower bound on the mean-squared-error (MSE) of estimators that satisfy some unbiasedness conditions. Although the CCRB unbiasedness conditions are satisfied asymptotically by the constrained maximum likelihood (CML) estimator, in the non-asymptotic region these conditions are usually too strict and the commonly-used estimators, such as the CML estimator, do not satisfy them. Therefore, the CCRB may not be a lower bound on the MSE matrix of such estimators. In this paper, we propose a new definition for unbiasedness under constraints, denoted by C-unbiasedness, which is based on using Lehmann-unbiasedness with a weighted MSE (WMSE) risk and taking into account the parametric constraints. In addition to C-unbiasedness, a Cram$\acute{\text{e}}$r-Rao-type bound on the WMSE of C-unbiased estimators, denoted as Lehmann-unbiased CCRB (LU-CCRB), is derived. This bound is a scalar bound that depends on the chosen weighted combination of estimation errors. It is shown that C-unbiasedness is less restrictive than the CCRB unbiasedness conditions. Thus, the set of estimators that satisfy the CCRB unbiasedness conditions is a subset of the set of C-unbiased estimators and the proposed LU-CCRB may be an informative lower bound in cases where the corresponding CCRB is not. In the simulations, we examine linear and nonlinear estimation problems under nonlinear constraints in which the CML estimator is shown to be C-unbiased and the LU-CCRB is an informative lower bound on the WMSE, while the corresponding CCRB on the WMSE is not a lower bound and is not informative in the non-asymptotic region.
\end{abstract}

\section{Introduction} \label{sec:Introduction}
In the non-Bayesian framework, the Cram$\acute{\text{e}}$r-Rao bound (CRB) \cite{Rao_paper,crlb},\cite{KAY} provides a lower bound on the mean-squared-error (MSE) matrix of any mean-unbiased estimator and is used as a benchmark for parameter estimation performance analysis. In some cases, scalar risks for multi-parameter estimation are of interest, for example due to tractability or complexity issues. Corresponding Cram$\acute{\text{e}}$r-Rao-type bounds for this case, can be found in e.g. \cite{HERO_UNIFORM_CRB,UNIFORM_ELDAR,YONINA}. In constrained parameter estimation \cite{Hero_constraint}, the unknown parameter vector satisfies given parametric constraints. In some cases, the CRB for constrained parameter estimation can be obtained by a reparameterization of the original problem. However, this approach may be intractable and may hinder insights into the original unconstrained problem \cite{Stoica_Ng}. In addition, mean-unbiased estimators may not exist for the reparameterized problem, as occurs in cases where the resulting distribution is periodic \cite{TODROS_WINNIK}, \cite{PHASE_KAY}.\\
\indent
In the pioneering work in \cite{Hero_constraint}, the constrained CRB (CCRB) was derived for constrained parameter estimation without reparameterizing the original problem. A simplified derivation of the CCRB was presented in \cite{Marzetta}. The CCRB was extended for various cases, such as parameter estimation with a singular Fisher information matrix (FIM) in \cite{Stoica_Ng}, complex vector parameter estimation in \cite{CCRB_complex}, biased estimation in \cite{BenHaim}, and sparse parameter vector estimation in \cite{sparse_con}. Alternative derivations of the CCRB from a model fitting perspective and via norm minimization were presented in \cite{Moore_fitting} and \cite{normC}, respectively. A hybrid Bayesian and non-Bayesian CCRB and the CCRB under misspecified models were derived in \cite{CCRB_hybrid} and \cite{CCRB_misspecified}, respectively. Computations of the CCRB in various applications can be found, for example, in \cite{HERO_EMISSION,SADLER_sim,WIJNHOLDS_sim,ROUTTENBERG_sim,HERO_HYPER,MENNI_sim}. In addition to the CCRB, Cram$\acute{\text{e}}$r-Rao-type bounds for estimation of parameters constrained to lie on a manifold were derived in \cite{HENDRIKS,XAVIER_BARROSO,SMITH,BOUMAL}. The constrained Bhattacharyya bound was derived in \cite{HERO_ICASSP} and the constrained Hammersley-Chapman-Robbins (HCR) bound was derived in \cite{Hero_constraint} by using the classical HCR bound in which the test-points were taken from the constrained set.\\
\indent
A popular estimator for constrained parameter estimation is the constrained maximum likelihood (CML) estimator \cite{Marzetta,Moore_fitting,AITCHISON_SILVEY,SILVEY_LAGRANGE,CROWDER,OSBORNE,LUO_BOUCHARD,Moore_phd,Moore_scoring}. This estimator is obtained by maximizing the likelihood function subject to parametric constraints. It is shown in \cite{Marzetta,Moore_fitting} for nonsingular and singular FIM, respectively, that if there exists a mean-unbiased estimator satisfying the constraints that achieves the CCRB, then this estimator is a stationary point of the constrained likelihood maximization. Asymptotic properties of the CML estimator can be found in \cite{AITCHISON_SILVEY,SILVEY_LAGRANGE,CROWDER,OSBORNE,LUO_BOUCHARD,Moore_phd,Moore_scoring}, under different assumptions. In particular, under mild assumptions, the CML estimator asymptotically satisfies the CCRB unbiasedness conditions and attains the CCRB for both linear and nonlinear constraints, as shown in \cite{OSBORNE} and \cite{Moore_scoring}, respectively. However, in the non-asymptotic region the CML estimator may not satisfy the CCRB unbiasedness conditions \cite{SOMEKH_LESHEM,SSP_EYAL} and therefore, the CCRB may not be an informative lower bound for CML performance in the non-asymptotic region. Other estimation methods for constrained parameter estimation are based on minimax criteria, e.g. \cite{PILZ,ELDAR_ROBUST}, and least squares criteria, e.g. \cite{Moore_scoring,STOICA_LINEAR,UHLICH}.\\
\indent
It is well known that unrestricted minimization of the non-Bayesian MSE yields the trivial, parameter-dependent estimator. In order to avoid this loophole, mean-unbiasedness of estimators is usually imposed \cite{point_est,TTB1}, {\it i.e.} only estimators with zero bias are considered. In early works on constrained parameter estimation \cite{Stoica_Ng,Marzetta,Moore_fitting,Moore_scoring}, the CCRB was assumed to be a lower bound for estimators that satisfy the constraints and have zero bias in the constrained set. It was shown in \cite{SOMEKH_LESHEM} that zero-bias requirement may be too strict. In addition, it was shown (e.g. \cite{BenHaim,Moore_phd}) that the CCRB can be derived without requiring the estimator to satisfy the constraints. The unbiasedness conditions of the CCRB were thoroughly discussed in \cite{sparse_con} and were shown to be less restrictive than the unbiasedness conditions of the conventional CRB. However, the CCRB unbiasedness conditions may still be too strict for commonly-used estimators, such as the CML estimator.\\
\indent
In this paper, the concept of unbiasedness in the Lehmann sense under parametric constraints, named C-unbiasedness, is developed. The Lehmann-unbiasedness \cite{point_est,LEHMANN_CONCEPT} generalizes the mean-unbiasedness to arbitrary cost functions and arbitrary parameter space. It has been used in various works for derivation of performance bounds under different cost functions \cite{PCRB,BAR1,CYCLIC,SELECTION}. Using the C-unbiasedness concept, we derive a new constrained Cram$\acute{\text{e}}$r-Rao-type lower bound, named Lehmann-unbiased CCRB (LU-CCRB), on the weighted MSE (WMSE) \cite{EXTENDED_ZZ,PILZ,ELDAR_WEIGHTED_MSE,GRUBER} of any C-unbiased estimator. It is shown that for linear constraints and/or in the asymptotic region, the proposed LU-CCRB coincides with the corresponding CCRB. In the simulations, the CML estimator is shown to be C-unbiased for orthogonal linear estimation problem under norm constraint and for complex amplitude estimation with amplitude constraint and unknown frequency. Therefore, the LU-CCRB is a lower bound for CML performance in these cases. In contrast, the corresponding CCRB on the WMSE is not a lower bound in the considered cases, in the non-asymptotic region, and is shown to be significantly higher than the WMSE of the CML estimator. These results demonstrate that the LU-CCRB provides an informative WMSE lower bound in cases, where the corresponding CCRB on the WMSE, and consequently also the matrix CCRB, are not lower bounds.\\
\indent
The WMSE is a scalar risk for multi-parameter estimation that allows the consideration of any weighted sum of squared linear combinations of the estimation errors. In particular, the MSE matrix trace is a special case of the WMSE. Unlike the CCRB, which is a matrix lower bound, the proposed LU-CCRB is a family of scalar bounds, which provides a different lower bound for each weighted sum of squared linear combinations of the estimation errors under corresponding C-unbiasedness condition. An early derivation of C-unbiasedness and lower bounds on a projected MSE matrix appear in the conference paper \cite{SAM2012constraint}. In this work, we focus on WMSE rather than the projected MSE.\\
\indent
The remainder of the paper is organized as follows. In Section \ref{sec:Notations and background}, we define the notations and present relevant background for this paper. The C-unbiasedness and the LU-CCRB are derived in Sections \ref{sec:Unbiasedness under constraints} and \ref{sec:LU-CCRB}, respectively. Our simulations appear in Section \ref{sec:Examples}. In Section \ref{sec:Conclusion}, we give our conclusions.

\section{Notations and background} \label{sec:Notations and background}
\subsection{Notations and constrained model}
Throughout this paper, we denote vectors by boldface lowercase letters and matrices by boldface uppercase letters. The $m$th element of the vector $\avec$ and the $(m,k)$th element of the matrix $\Amat$ are denoted by $a_m$ and $[\Amat]_{m,k}$, respectively. A subvector of $\avec$ with indices $l,l+1,\ldots,l+K$ is denoted by $[\avec]_{l:l+K}$. The identity matrix of dimension $K\times K$ is denoted by $\Imat_K$ and $\zerovec$ denotes a vector/matrix of zeros. The notations ${\rm{Tr}}(\cdot)$ and ${\rm{vec}}(\cdot)$ denote the trace and vectorization operators, where the vectorization operator stacks the columns of its input matrix into a column vector. The notations $(\cdot)^T$, $(\cdot)^{-1}$, and $(\cdot)^{\dagger}$ denote the transpose, inverse, and Moore-Penrose pseudo-inverse, respectively. The notation $\Amat\succeq\Bmat$ implies that $\Amat-\Bmat$ is a positive semidefinite matrix. The column and null spaces of a matrix are denoted by ${\mathcal{R}}(\cdot)$ and ${\mathcal{N}}(\cdot)$, respectively. The matrices $\Pmat_\Amat=\Amat\Amat^{\dagger}=\Amat(\Amat^T\Amat)^{\dagger}\Amat^T$ and $\Pmat_\Amat^\bot=\Imat_M-\Pmat_\Amat$ are the orthogonal projection matrices onto ${\mathcal{R}}(\Amat)$ and ${\mathcal{N}}(\Amat^T)$, respectively \cite{CAMPBELL}. The notation $\Amat\otimes\Bmat$ is the Kronecker product of the matrices $\Amat$ and $\Bmat$. The gradient of a vector function $\gvec$ of $\thetavec$, $\nabla_{\thetavecsmall}\gvec(\thetavec)$, is a matrix in which $[\nabla_{\thetavecsmall}\gvec(\thetavec)]_{m,k}=\frac{{\partial}g_m(\thetavecsmall)}{{\partial}\theta_k}$. The real and imaginary parts of an argument are denoted by $\text{Re}\{\cdot\}$ and $\text{Im}\{\cdot\}$, respectively, and $j\define{\sqrt{-1}}$. The notation $\angle{\cdot}$ stands for the phase of a complex scalar, which is assumed to be restricted to the interval $[-\pi,\pi)$.\\
\indent
Let $(\Omega_\xvec,{\cal{F}},P_\thetavecsmall)$ denote a probability space, where $\Omega_\xvec$ is the observation space, ${\cal{F}}$ is the $\sigma$-algebra on $\Omega_\xvec$, and $\left\{P_\thetavecsmall \right\}$ is a family of probability measures parameterized by the deterministic unknown parameter vector $\thetavec \in {\mathbb{R}}^M$. Each probability measure $P_\thetavecsmall$ is assumed to have an associated probability density function (pdf), $f_\xvec(\cdot;\thetavec)$, such that
the expected value of any measurable function $Z:\Omega_\xvec\rightarrow {\mathbb{R}}$  with respect to (w.r.t.) $P_\thetavecsmall$ satisfies ${\rm{E}}[Z(\xvec);\thetavec]\define \int_{\Omega_\xvec} Z(\gammavec) f_\xvec(\gammavec;\thetavec){\ud} \gammavec$. For simplicity of notations, we omit $\thetavec$ from the notation of expectation and denote it by ${\rm{E}}[\cdot]$, whenever the value of $\thetavec$ is clear from the context. The conditional expectation given event $\mathcal{A}$ and parameterized by $\thetavec$ is denoted by ${\rm{E}}[\cdot|\mathcal{A};\thetavec]$.\\
\indent
We suppose that $\thetavec$ is restricted to the set
\be \label{set_def}
\Theta_\fvec=\{\thetavec\in{\mathbb{R}}^M:\fvec(\thetavec)=\zerovec\},
\ee
where $\fvec:{\mathbb{R}}^M\rightarrow {\mathbb{R}}^K $ is a continuously differentiable function. It is assumed that $0\leq K<M$ and that the matrix, $\Fmat(\thetavec)=\nabla_{\thetavecsmall}\fvec(\thetavec)\in{\mathbb{R}}^{K\times M}$,  
has full row rank for any $\thetavec\in\Theta_\fvec$, {\it i.e.} the constraints are not redundant. Thus, for any $\thetavec\in\Theta_\fvec$ there exists a matrix $\Umat(\thetavec)\in{\mathbb{R}}^{M\times(M-K)}$, such that
\be \label{one}
\Fmat(\thetavec)\Umat(\thetavec)=\zerovec
\ee
and
\be \label{two}
\Umat^T(\thetavec)\Umat(\thetavec)=\Imat_{M-K}.
\ee
The case $K=0$ implies an unconstrained estimation problem in which $\Umat(\thetavec)=\Imat_M$. Under the assumption that each element of $\Umat(\thetavec)$ is differentiable w.r.t. $\thetavec,~\forall \thetavec\in\Theta_\fvec$, we define
\be\label{Vmat_define}
\Vmat_m(\thetavec)\define\nabla_{\thetavecsmall}\uvec_m(\thetavec),
\ee
where $\uvec_m(\thetavec)$ is the $m$th column of $\Umat(\thetavec),~\forall m=1,\ldots,M-K$.\\
\indent
An estimator of $\thetavec$ based on a random observation vector $\xvec\in\Omega_\xvec$ is denoted by $\hat{\thetavec}:\Omega_\xvec\rightarrow{\mathbb{R}}^M$,  where $\hat{\thetavec}(\xvec)$ does not necessarily satisfy the constraints. For the sake of simplicity, in the following $\hat{\thetavec}(\xvec)$ is replaced by $\hat{\thetavec}$. The bias of an estimator is denoted by
\be\label{bias_definition}
\bvec_{\hat{\thetavecsmall}}(\thetavec)\define{\rm{E}}[\hat{\thetavec}-\thetavec].
\ee
Under the assumption that each element of $\bvec_{\hat{\thetavecsmall}}(\thetavec)$ is differentiable w.r.t. $\thetavec,~\forall \thetavec\in\Theta_\fvec$, we define the bias gradient 
\be\label{bias_gradient}
\Dmat_{\hat{\thetavecsmall}}(\thetavec)\define\nabla_{\thetavecsmall}\bvec_{\hat{\thetavecsmall}}(\thetavec).
\ee

\subsection{WMSE and CCRB}\label{subsec:WMSE and CCRB}
In this paper, we are interested in estimation under a weighted squared-error (WSE) cost function \cite{EXTENDED_ZZ,PILZ,ELDAR_WEIGHTED_MSE,GRUBER},
\begin{equation}\label{WSE}
C_{\text{WSE}}(\hat{\thetavec},\thetavec)=(\hat{\thetavec}-\thetavec)^T\Wmat(\hat{\thetavec}-\thetavec),
\end{equation}
where $\Wmat\in\mathbb{R}^{M\times M}$ is a positive semidefinite weighting matrix. The WMSE risk is obtained by taking the expectation of \eqref{WSE} and is given by
\be\label{WMSE}
{\text{WMSE}}_{\hat{\thetavecsmall}}(\thetavec)\define{\rm{E}}\left[C_{\text{WSE}}(\hat{\thetavec},\thetavec)\right]={\rm{E}}\left[(\hat{\thetavec}-\thetavec)^T\Wmat(\hat{\thetavec}-\thetavec)\right].
\ee
The WMSE is in fact a family of scalar risks for estimation of an unknown parameter vector, where for each $\Wmat$ we obtain a different risk. Therefore, the WMSE allows flexibility in the design of estimators and the derivation of performance bounds. For example, by choosing $\Wmat=\Imat_M$ we obtain the special case of the MSE matrix trace. Another example is when one may wish to consider the estimation of each element of the unknown parameter vector separately. Moreover, $\Wmat$ can compensate for possibly different units of the parameter vector elements. Another example is estimation in the presence of nuisance parameters, where we are only interested in the MSE for estimation of a subvector of the unknown parameter vector (see e.g. \cite{SSP_EYAL}, \cite[p. 461]{point_est}) and thus, $\Wmat$ includes zero elements for the nuisance parameters.\\
\indent
Let 
\be\label{l_define}
\upsilonvec(\xvec,\thetavec)\define\nabla_{\thetavecsmall}^T\log f_\xvec(\xvec;\thetavec)
\ee
and the FIM
\be\label{FIM}
\Jmat(\thetavec)\define{\rm{E}}\left[\upsilonvec(\xvec,\thetavec)\upsilonvec^T(\xvec,\thetavec)\right].
\ee
At $\thetavec_0\in\Theta_\fvec$, under the assumption
\begin{multline}\label{yonina_assume}
\mathcal{R}\left(\Umat(\thetavec_0)\Umat^T(\thetavec_0)\right)\\
\subseteq\mathcal{R}\left(\Umat(\thetavec_0)\Umat^T(\thetavec_0)\Jmat(\thetavec_0)\Umat(\thetavec_0)\Umat^T(\thetavec_0)\right),
\end{multline}
the CCRB is given by \cite{Stoica_Ng,BenHaim}
\be \label{CCRBmat}
\Bmat_{\text{CCRB}}(\thetavec_0)\define\Umat(\thetavec_0)\left(\Umat^T(\thetavec_0)\Jmat(\thetavec_0)\Umat(\thetavec_0)\right)^{\dagger}\Umat^T(\thetavec_0).
\ee
The CCRB is an MSE matrix lower bound that can be reformulated as a WMSE lower bound by multiplying the bound by the weighting matrix, $\Wmat$, and taking the trace. That is, based on the matrix CCRB from \eqref{CCRBmat} we obtain the following WMSE lower bound
\begin{multline} \label{W_CCRB_W}
B_{\text{CCRB}}(\thetavec_0,\Wmat)\define{\text{Tr}}\left(\Bmat_{\text{CCRB}}(\thetavec_0)\Wmat\right)\\
={\text{Tr}}\left(\left(\Umat^T(\thetavec_0)\Jmat(\thetavec_0)\Umat(\thetavec_0)\right)^{\dagger}\left(\Umat^T(\thetavec_0)\Wmat\Umat(\thetavec_0)\right)\right).
\end{multline}
Computations of the CCRB on the WMSE with different weighting matrices can be found in e.g. \cite{Stoica_Ng,Moore_scoring,SSP_EYAL}. In the following, we refer to the WMSE lower bound in \eqref{W_CCRB_W} as the CCRB for the considered choice of weighting matrix.\\
\indent
It is known that the CRB is a local bound, which is a lower bound for estimators whose bias and bias gradient vanish at a considered point (see e.g. \cite{sparse_con}), that is, locally mean-unbiased estimators in the vicinity of this point. In \cite{sparse_con}, local $\mathcal{X}$-unbiasedness is defined as follows:
\begin{definition}
\label{X_unbiased_definition}
The estimator $\hat{\thetavec}$ is said to be a locally $\mathcal{X}$-unbiased estimator in the vicinity of $\thetavec_0\in\Theta_\fvec$ if it satisfies
\be \label{point_wise_chi_unb}
\bvec_{\hat{\thetavecsmall}}(\thetavec_0)=\zerovec
\ee
and
\be \label{local_wise_chi_unb}
\Dmat_{\hat{\thetavecsmall}}(\thetavec_0)\Umat(\thetavec_0)=\zerovec.
\ee
\end{definition}
It is shown in \cite{sparse_con} that the CCRB is a lower bound for locally $\mathcal{X}$-unbiased estimators, where local $\mathcal{X}$-unbiasedness is a weaker restriction than local mean-unbiasedness. As a result, the CCRB is always lower than or equal to the CRB. In the following section, we derive a different unbiasedness definition for constrained parameter estimation, named C-unbiasedness, whose local definition is less restrictive than local $\mathcal{X}$-unbiasedness.

\section{Unbiasedness under constraints} \label{sec:Unbiasedness under constraints}
In non-Bayesian parameter estimation, direct minimization of the risk w.r.t. the estimator results in a trivial estimator. Accordingly, one needs to exclude such estimators by additional restrictions on the considered set of estimators. A common restriction on estimators is mean-unbiasedness, which is used for derivation of the CRB. In the following, we propose a novel unbiasedness restriction for constrained parameter estimation, named C-unbiasedness, which is based on Lehmann's definition of unbiasedness. It is shown that local C-unbiasedness is a weaker restriction than the local unbiasedness restrictions of the CCRB. Therefore, local C-unbiasedness allows for a larger set of estimators to be considered.

\subsection{Lehmann-unbiasedness} \label{subsec:Lehmann-unbiasedness}
Lehmann \cite{point_est,LEHMANN_CONCEPT} proposed a generalization of the unbiasedness concept based on the considered cost function and parameter space, as presented in the following definition.
\begin{definition}
\label{unbiased_definition}
The estimator $\hat{\thetavec}$ is said to be a uniformly unbiased estimator of $\thetavec\in\Omega_\thetavecsmall$ in the Lehmann sense \cite{point_est,LEHMANN_CONCEPT} w.r.t. the cost function $C(\hat{\thetavec},\thetavec)$ 
if 
\be \label{defdef}
{\rm{E}}[C(\hat{\thetavec},\etavec);\thetavec] \geq {\rm{E}}[C(\hat{\thetavec},\thetavec);\thetavec],~\forall \etavec,\thetavec\in \Omega_\thetavecsmall,
\ee
where $\Omega_\thetavecsmall$ is the parameter space.
\end{definition}
The Lehmann-unbiasedness definition implies that an estimator is unbiased if on the average it is ``closest'' to the true parameter, $\thetavec$, rather than to any other value in the parameter space, $\etavec\in\Omega_\thetavecsmall$. The measure of closeness between the estimator and the parameter is the cost function, $C(\hat{\thetavec},\thetavec)$. For example, in \cite{LEHMANN_CONCEPT} it is shown that under the scalar squared-error cost function, $C(\hat{\theta},\theta)=(\hat{\theta}-\theta)^2,~\theta\in\mathbb{R}$, the Lehmann-unbiasedness  in \eqref{defdef} is reduced to the conventional mean-unbiasedness, ${\rm{E}}[\hat{\theta}-\theta]=0$, $\forall\theta\in\mathbb{R}$. Lehmann-unbiasedness conditions for various cost functions can be found in \cite{PCRB,BAR1,CYCLIC,SELECTION}.\\
\indent
In non-Bayesian estimation theory, two types of unbiasedness are usually considered: uniform unbiasedness in which the estimator is unbiased at any point in the parameter space and local unbiasedness (see e.g. \cite{localU}) in which the estimator is assumed to be unbiased only in the vicinity of the true parameter $\thetavec_0$. In the following definition, we extend the original uniform definition of Lehmann-unbiasedness in \eqref{defdef} to local Lehmann-unbiasedness.
\begin{definition}\label{Def_local_Cunb}
\label{Loc_unbiased_definition}
The estimator $\hat{\thetavec}$ is said to be a locally Lehmann-unbiased estimator in the vicinity of $\thetavec_0\in\Omega_\thetavecsmall$ w.r.t. the cost function $C(\hat{\thetavec},\thetavec)$ 
if 
\be \label{Loc_defdef2}
{\rm{E}}[C(\hat{\thetavec},\etavec);\thetavec] \geq {\rm{E}}[C(\hat{\thetavec},\thetavec);\thetavec],~\forall \etavec\in \Omega_\thetavecsmall,
\ee
for any $\thetavec\in\Omega_\thetavecsmall$, s.t. $|\theta_m-\theta_{0,m}|<\varepsilon_m,~\varepsilon_m\to0,~\forall m=1,\ldots,M$.
\end{definition}

\subsection{Uniform C-unbiasedness} \label{subsec:Uniform C-unbiasedness}
In the following, the uniform C-unbiasedness is derived by combining the uniform Lehmann-unbiasedness condition from \eqref{defdef} w.r.t. the WSE cost function and the parametric constraints.
\begin{proposition} \label{Cunbias_prop}
A necessary condition for an estimator $\hat{\thetavec}:\Omega_\xvec\rightarrow{\mathbb{R}}^M$
to be a uniformly unbiased estimator of $\thetavec\in{\mathbb{R}}^M$ in the Lehmann sense w.r.t. the WSE cost function and under the constrained set in \eqref{set_def} is
\be \label{unbiased_cond_nec}
\Umat^T(\thetavec)\Wmat \bvec_{\hat{\thetavecsmall}}(\thetavec)=\zerovec,~\forall \thetavec\in\Theta_\fvec.
\ee
\end{proposition} 
\begin{proof}
By substituting $\Omega_\thetavecsmall=\Theta_\fvec$ and the WSE cost function from \eqref{WSE} in \eqref{defdef}, one obtains that Lehmann-unbiasedness in the constrained setting is reduced to:
\begin{multline} \label{const_L}
{\rm{E}}\left[(\hat{\thetavec}-\etavec)^T\Wmat(\hat{\thetavec}-\etavec);\thetavec\right]\\
\geq{\rm{E}}\left[(\hat{\thetavec}-\thetavec)^T\Wmat(\hat{\thetavec}-\thetavec);\thetavec\right],~\forall \etavec,\thetavec\in \Theta_\fvec.
\end{multline}
The condition in \eqref{const_L} is equivalent to requiring $\etavec_{min}=\thetavec$, where $\etavec_{min}$ is the minimizer of the following constrained minimization problem
\be \label{minQ}
\underset{\etavecsmall}{\min}~ {\rm{E}}[(\hat{\thetavec}-\etavec)^T\Wmat(\hat{\thetavec}-\etavec);\thetavec]~~\text{s.t.}~~\fvec\left(\etavec\right)=\zerovec.
\ee
By using a necessary condition for constrained minimization (see e.g. Eq. (1.62) in \cite{BETTS}), it can be shown that the minimizer of \eqref{minQ}, $\etavec_{min}$, must satisfy
\be \label{unbiased_cond_nec_proof_book}
\Umat^T(\etavec_{min})\left.\nabla_{\etavecsmall}^T{\rm{E}}[(\hat{\thetavec}-\etavec)^T\Wmat(\hat{\thetavec}-\etavec);\thetavec]\right|_{\etavecsmall=\etavecsmall_{min}}=\zerovec.
\ee
Under the assumption that integration w.r.t. $\xvec\in\Omega_\xvec$ and derivatives
w.r.t. $\etavec$ can be reordered, the condition in \eqref{unbiased_cond_nec_proof_book} is equivalent to
\be \label{unbiased_cond_nec_proof}
\Umat^T(\etavec_{min})\Wmat {\rm{E}}[\hat{\thetavec}-\etavec_{min};\thetavec]=\zerovec.
\ee
Finally, by substituting $\etavec_{min}=\thetavec$ and \eqref{bias_definition} in \eqref{unbiased_cond_nec_proof}, one obtains \eqref{unbiased_cond_nec}.
\end{proof}
An estimator that satisfies \eqref{unbiased_cond_nec} is said to be uniformly C-unbiased. The uniform C-unbiasedness is a necessary condition for uniform Lehmann-unbiasedness w.r.t. the WSE cost function and under the constrained set in \eqref{set_def}. It can be seen that if an estimator has zero mean-bias in the constrained set, {\it i.e.} $\bvec_{\hat{\thetavecsmall}}(\thetavec)=\zerovec$, $\forall\thetavec\in \Theta_\fvec$, then it satisfies \eqref{unbiased_cond_nec} but not vice versa. Thus, the uniform C-unbiasedness condition is a weaker condition than requiring mean-unbiasedness in the constrained set.

\subsection{Local C-unbiasedness} \label{subsec:Local C-unbiasedness}
In this subsection, local C-unbiasedness conditions are derived by combining the local Lehmann-unbiasedness condition from \eqref{Loc_defdef2} w.r.t. the WSE cost function and the parametric constraints.
\begin{proposition} \label{LocCunbias_prop}
Necessary conditions for an estimator $\hat{\thetavec}:\Omega_\xvec\rightarrow{\mathbb{R}}^M$
to be a locally Lehmann-unbiased estimator in the vicinity of $\thetavec_0\in\Theta_\fvec$ w.r.t. the WSE cost function and under the constrained set in \eqref{set_def} are
\be \label{point_wise_unb}
\Umat^T(\thetavec_0)\Wmat\bvec_{\hat{\thetavecsmall}}(\thetavec_0)=\zerovec
\ee
and
\be\label{sec_localU_final}
\bvec_{\hat{\thetavecsmall}}^T(\thetavec_0)\Wmat\Vmat_m(\thetavec_0)\Umat(\thetavec_0)=-\uvec_m^T(\thetavec_0)\Wmat\Dmat_{\hat{\thetavecsmall}}(\thetavec_0)\Umat(\thetavec_0),
\ee
$\forall m=1,\ldots,M-K$, where $\Vmat_m(\thetavec)$ and $\Dmat_{\hat{\thetavecsmall}}(\thetavec_0)$ are defined in \eqref{Vmat_define} and \eqref{bias_gradient}, respectively.
\end{proposition} 
\begin{proof}
The proof is given in Appendix \ref{App_LocCunbias_prop}.
\end{proof}
\noindent In particular, in case the MSE matrix trace is of interest, we substitute $\Wmat=\Imat_M$ in \eqref{point_wise_unb}-\eqref{sec_localU_final} and the resulting local C-unbiasedness conditions are 
\be \label{point_wise_unb_I}
\Umat^T(\thetavec_0)\bvec_{\hat{\thetavecsmall}}(\thetavec_0)=\zerovec
\ee
and
\be\label{sec_localU_final_I}
\bvec_{\hat{\thetavecsmall}}^T(\thetavec_0)\Vmat_m(\thetavec_0)\Umat(\thetavec_0)=-\uvec_m^T(\thetavec_0)\Dmat_{\hat{\thetavecsmall}}(\thetavec_0)\Umat(\thetavec_0),
\ee
$\forall m=1,\ldots,M-K$.\\
\indent
For any positive semidefinite matrix $\Wmat$, it can be seen that if an estimator satisfies \eqref{point_wise_chi_unb}-\eqref{local_wise_chi_unb}, then it satisfies also \eqref{point_wise_unb}-\eqref{sec_localU_final} but not vice versa. Thus, for any positive semidefinite weighting matrix $\Wmat$, the local C-unbiasedness is a weaker restriction than the local $\mathcal{X}$-unbiasedness and therefore, lower bounds on the WMSE of locally C-unbiased estimators may be lower than the corresponding CCRB. In Section \ref{sec:Examples}, we show examples in which the CML estimator is C-unbiased and is not $\mathcal{X}$-unbiased. In case that some of the elements of $\thetavec$ are considered as nuisance parameters, we can put zero weights on these elements in the weighting matrix $\Wmat$. It can be seen that in this case, the local C-unbiasedness conditions from \eqref{point_wise_unb}-\eqref{sec_localU_final} are not affected by the bias function of a nuisance parameter estimator. 

\section{LU-CCRB} \label{sec:LU-CCRB}
In this section, we derive the LU-CCRB, which is a new Cram$\acute{\text{e}}$r-Rao-type lower bound on the WMSE of locally C-unbiased estimators, where the WMSE is defined in \eqref{WMSE}. Properties of this bound are described in Subsection \ref{subsec:Properties of LU-CCRB}.

\subsection{Derivation of LU-CCRB} \label{subsec:Derivation of LU-CCRB}
In the following theorem, we derive the LU-CCRB on the WMSE of locally C-unbiased estimators. For the derivation we define $\Cmat_{\Umat,\Wmat}(\thetavec)\in\mathbb{R}^{(M-K)^2\times (M-K)^2}$, which is a block matrix whose $(m,k)$th block is given by
\be\label{Gamma_block_define}
\Cmat_{\Umat,\Wmat}^{(m,k)}(\thetavec)\define\left(\Smat^{(m)}_{\Wmat}(\thetavec)\right)^T\Wmat\Smat^{(k)}_{\Wmat}(\thetavec),
\ee
$\forall m,k=1,\ldots,M-K$, where
\be\label{S_define}
\begin{split}
\Smat^{(m)}_{\Wmat}(\thetavec)\define\Wmat^{\dagger}\Wmat^{\frac{1}{2}}\Pmat_{\Wmat^{\frac{1}{2}}\Umat}^\bot(\thetavec)\Wmat^{\frac{1}{2}}\Vmat_m(\thetavec)\Umat(\thetavec)
\end{split}
\ee
is an $M\times(M-K)$ matrix, $\forall m=1,\ldots,M-K$. Finally, we define the matrix
\begin{multline}\label{Gammamat_define}
\Gammamat_{\Umat,\Wmat}(\thetavec)\define\Cmat_{\Umat,\Wmat}(\thetavec)\\
+\left(\Umat^T(\thetavec)\Wmat\Umat(\thetavec)\right)\otimes\left(\Umat^T(\thetavec)\Jmat(\thetavec)\Umat(\thetavec)\right).
\end{multline}
In the following theorem, we present the proposed LU-CCRB on the WMSE of locally C-unbiased estimators, where the local C-unbiasedness conditions are presented in \eqref{point_wise_unb}-\eqref{sec_localU_final}. The LU-CCRB is a family of scalar bounds, which provides a different lower bound for each choice of weighting matrix.
\begin{Theorem} \label{T3}
Let $\hat{\thetavec}$ be a locally C-unbiased estimator of $\thetavec$ in the vicinity of $\thetavec_0\in\Theta_\fvec$ for a given positive semidefinite weighting matrix $\Wmat$ and assume 
\renewcommand{\theenumi}{C.\arabic{enumi}} 
\begin{enumerate}
\setcounter{enumi}{0}
\item \label{cond1CRB} 
Integration w.r.t. $\xvec$ and differentiation w.r.t. $\thetavec$ at $\thetavec_0\in\Theta_\fvec$ can be interchanged.
\item \label{cond2CRB} 
Each element of $\Umat(\thetavec)$ is differentiable w.r.t. $\thetavec$ at $\thetavec_0\in\Theta_\fvec$.
\item \label{cond3CRB} 
$\Jmat(\thetavec_0)$ is finite.
\end{enumerate} 
\renewcommand{\theenumi}{\arabic{enumi}}
Then,
\be \label{total_CR_bound_CS_pre_W}
{\text{WMSE}}_{\hat{\thetavecsmall}}(\thetavec_0)\geq B_{\text{LU-CCRB}}(\thetavec_0,\Wmat),
\ee
where
\be \label{total_CR_bound_CS}
\begin{split}
B_{\text{LU-CCRB}}(\thetavec_0,\Wmat)&\define {\rm{vec}}^T\left(\Umat^T(\thetavec_0)\Wmat\Umat(\thetavec_0)\right)\Gammamat_{\Umat,\Wmat}^{\dagger}(\thetavec_0)\\
&~~~\times{\rm{vec}}\left(\Umat^T(\thetavec_0)\Wmat\Umat(\thetavec_0)\right).
\end{split}
\ee
Equality in \eqref{total_CR_bound_CS_pre_W} is obtained {\em iff} 
\be \label{equality_cond_Prop_W}
\begin{split}
&\Wmat^{\frac{1}{2}}(\hat{\thetavec}-\thetavec_0)\\
&=\Wmat^{\frac{1}{2}}\sum_{m=1}^{M-K}\left(\left(\Smat^{(m)}_{\Wmat}(\thetavec_0)+\right.\left.\Tmat^{(m)}_{\Wmat}(\xvec,\thetavec_0)\right)\left[\Gammamat_{\Umat,\Wmat}^{\dagger}(\thetavec_0)\right.\right.\\
~~&\left.\left.\times{\rm{vec}}\left(\Umat^T(\thetavec_0)\Wmat\Umat(\thetavec_0)\right)\right]_{((m-1)(M-K)+1):(m(M-K))}\right),
\end{split}
\ee
where
\begin{equation}\label{T_define}
\Tmat^{(m)}_{\Wmat}(\xvec,\thetavec)\define\Wmat^{\dagger}\Wmat\uvec_m(\thetavec)\upsilonvec^T(\xvec,\thetavec)\Umat(\thetavec),
\end{equation}
$\forall m=1,\ldots,M-K$.
\end{Theorem}
\begin{proof}
The proof is given in Appendix \ref{App_T3}.
\end{proof}
It can be seen that computation of the CCRB requires the evaluation of $\Umat(\thetavec)$ \cite{Moore_scoring}, while computation of the LU-CCRB requires the evaluation of $\Umat(\thetavec)$ and the matrices $\Vmat_m(\thetavec),~\forall m=1,\ldots,M-K$. The matrix $\Vmat_m(\thetavec)$ can be evaluated numerically by using $\Fmat(\thetavec)$ and $\Umat(\thetavec)$ and applying the product rule on \eqref{one} and \eqref{two}, $\forall m=1,\ldots,M-K$.\\
\indent		
In order to obtain a lower bound on the MSE matrix trace, we substitute $\Wmat=\Imat_M$ and \eqref{two}, in \eqref{total_CR_bound_CS} and obtain
\be \label{total_CR_bound_CS_I}
B_{\text{LU-CCRB}}(\thetavec_0,\Imat_M)\define{\rm{vec}}^T\left(\Imat_{M-K}\right)\Gammamat_{\Umat,\Imat_M}^{\dagger}(\thetavec_0){\rm{vec}}\left(\Imat_{M-K}\right),
\ee
where
\be\label{Gammamat_define_I}
\Gammamat_{\Umat,\Imat_M}(\thetavec)\define\Cmat_{\Umat,\Imat_M}(\thetavec)+\Imat_{M-K}\otimes\left(\Umat^T(\thetavec)\Jmat(\thetavec)\Umat(\thetavec)\right).
\ee
and
\be\label{Gamma_block_define_I}
\Cmat_{\Umat,\Imat_M}^{(m,k)}(\thetavec)\define\Umat^T(\thetavec)\Vmat_m^T(\thetavec)\Pmat_\Umat^\bot(\thetavec)\Vmat_k(\thetavec)\Umat(\thetavec),
\ee
$\forall m,k=1,\ldots,M-K$.

\subsection{Properties of LU-CCRB} \label{subsec:Properties of LU-CCRB}
\subsubsection{Relation to CCRB}
In the following proposition, we show the condition for $B_{\text{LU-CCRB}}(\thetavec_0,\Wmat)$ and $B_{\text{CCRB}}(\thetavec_0,\Wmat)$ to coincide.
\begin{proposition} \label{PROP_GENERAL}
Assume that $B_{\text{LU-CCRB}}(\thetavec_0,\Wmat)$ and $B_{\text{CCRB}}(\thetavec_0,\Wmat)$ exist and that
\be\label{zero_assume_W}
\Cmat_{\Umat,\Wmat}(\thetavec_0)=\zerovec,
\ee
Then,
\be\label{CCRB_zero_assume_equality_W}
B_{\text{LU-CCRB}}(\thetavec_0,\Wmat)=B_{\text{CCRB}}(\thetavec_0,\Wmat).
\ee
\end{proposition}
\begin{proof}
By substituting \eqref{zero_assume_W} in \eqref{Gammamat_define}, we obtain
\be\label{Su_Gammamat_define_assume}
\Gammamat_{\Umat,\Wmat}(\thetavec_0)=\left(\Umat^T(\thetavec_0)\Wmat\Umat(\thetavec_0)\right)\otimes\left(\Umat^T(\thetavec_0)\Jmat(\thetavec_0)\Umat(\thetavec_0)\right).
\ee 
Substituting \eqref{Su_Gammamat_define_assume} in \eqref{total_CR_bound_CS} and using the equality \cite[p. 22]{MATRIX_COOKBOOK}
\be\label{kronecker_pseudo}
(\Amat_1\otimes\Amat_2)^{\dagger}=\Amat_1^{\dagger}\otimes\Amat_2^{\dagger},
\ee
one obtains
\be \label{Su_LU_CCRB_zero_next}
\begin{split}
&B_{\text{LU-CCRB}}(\thetavec_0,\Wmat)\\
&={\rm{vec}}^T\left(\Umat^T(\thetavec_0)\Wmat\Umat(\thetavec_0)\right)\left(\left(\Umat^T(\thetavec_0)\Wmat\Umat(\thetavec_0)\right)^{\dagger}\right.\\
&~~~\left.\otimes\left(\Umat^T(\thetavec_0)\Jmat(\thetavec_0)\Umat(\thetavec_0)\right)^{\dagger}\right){\rm{vec}}\left(\Umat^T(\thetavec_0)\Wmat\Umat(\thetavec_0)\right).
\end{split}
\ee
By substituting the equality \cite[p. 60]{MATRIX_COOKBOOK}
\be\label{kronecker_prop}
{\rm{vec}}^T(\Amat_4)(\Amat_2^T\otimes\Amat_1){\rm{vec}}(\Amat_3)={\rm{Tr}}(\Amat_4^T\Amat_1\Amat_3\Amat_2)
\ee
with $\Amat_1=\left(\Umat^T(\thetavec_0)\Jmat(\thetavec_0)\Umat(\thetavec_0)\right)^{\dagger}$, $\Amat_2=\left(\Umat^T(\thetavec_0)\Wmat\Umat(\thetavec_0)\right)^{\dagger}$, and $\Amat_3=\Amat_4=\Umat^T(\thetavec_0)\Wmat\Umat(\thetavec_0)$ in \eqref{Su_LU_CCRB_zero_next}, we obtain
\be \label{Su_LU_CCRB_zero_second_step}
\begin{split}
&B_{\text{LU-CCRB}}(\thetavec_0,\Wmat)\\
&={\rm{Tr}}\left(\left(\Umat^T(\thetavec_0)\Wmat\Umat(\thetavec_0)\right)\left(\Umat^T(\thetavec_0)\Jmat(\thetavec_0)\Umat(\thetavec_0)\right)^{\dagger}\right.\\
&~~~\times\left.\left(\Umat^T(\thetavec_0)\Wmat\Umat(\thetavec_0)\right)\right.\left.\left(\Umat^T(\thetavec_0)\Wmat\Umat(\thetavec_0)\right)^{\dagger}\right)\\
&=B_{\text{CCRB}}(\thetavec_0,\Wmat),
\end{split}
\ee
where the second equality is obtained by using trace and pseudo-inverse properties and substituting \eqref{W_CCRB_W}.
\end{proof}
It can be shown that \eqref{zero_assume_W} is satisfied, for example, for linear constraints,
\be\label{linear}
\fvec(\thetavec)=\Amat\thetavec+\bvec=\zerovec,~\Amat\in\mathbb{R}^{K\times M},~\bvec\in\mathbb{R}^{K}. 
\ee
In this case, both the constraint gradient matrix, $\Fmat(\thetavec)=\Amat$, and the orthonormal null space matrix, $\Umat$, from \eqref{one}-\eqref{two}, are not functions of $\thetavec$. Therefore, the derivatives of the elements of $\Umat$ w.r.t. $\thetavec$ are zero, {\it i.e.} $\Vmat_m=\zerovec,~\forall m=1,\ldots,M-K$, and it can be verified by using \eqref{Gamma_block_define}-\eqref{S_define} that \eqref{zero_assume_W} is satisfied. Therefore, for linear constraints, the proposed LU-CCRB from \eqref{total_CR_bound_CS} coincides with the corresponding CCRB from \eqref{W_CCRB_W}. In particular, for linear Gaussian model under linear constraints, the CML is an $\mathcal{X}$-unbiased and C-unbiased estimator, and achieves the CCRB \cite{Moore_scoring} and the LU-CCRB.

\subsubsection{Order relation}
In the following proposition, we show that for the general case, the proposed LU-CCRB from \eqref{total_CR_bound_CS} is lower than or equal to the corresponding CCRB from \eqref{W_CCRB_W}.
\begin{proposition}\label{PROP_order}
Assume that $B_{\text{LU-CCRB}}(\thetavec_0,\Wmat)$ and $B_{\text{CCRB}}(\thetavec_0,\Wmat)$ exist and that \eqref{yonina_assume} holds. Then, 
\be\label{CCRB_is_higher_W}
B_{\text{CCRB}}(\thetavec_0,\Wmat)\geq B_{\text{LU-CCRB}}(\thetavec_0,\Wmat).
\ee
\end{proposition}
\begin{proof}
The proof is given in Appendix \ref{App_PROP_order}.	
\end{proof}
The LU-CCRB requires local C-unbiasedness and the CCRB requires local $\mathcal{X}$-unbiasedness, as mentioned in Subsections \ref{subsec:Derivation of LU-CCRB} and \ref{subsec:WMSE and CCRB}, respectively. The local C-unbiasedness is sufficient and less restrictive than local $\mathcal{X}$-unbiasedness, as mentioned in Subsection \ref{subsec:Local C-unbiasedness}. Therefore, the set of estimators for which the LU-CCRB is a lower bound contains the set of estimators for which the CCRB is a lower bound. This result elucidates the order relation in \eqref{CCRB_is_higher_W}. In Section \ref{sec:Examples}, we show examples in which the CML estimator is C-unbiased and is not $\mathcal{X}$-unbiased. As a result, in these examples the LU-CCRB is a lower bound on the WMSE of the CML estimator, while the CCRB on the WMSE is not necessarily a lower bound in the non-asymptotic region. The considered examples indicate that C-unbiasedness and the proposed LU-CCRB are more appropriate than $\mathcal{X}$-unbiasedness and the CCRB, respectively, for constrained parameter estimation.

\subsubsection{Asymptotic properties}\label{subsubsec:Asymptotic properties}
It is well known that under some conditions \cite{Moore_scoring}, the CCRB is attained asymptotically by the CML estimator. Consequently, the WMSE of the CML estimator asymptotically coincides with $B_{\text{CCRB}}(\thetavec_0,\Wmat)$ from \eqref{W_CCRB_W}. Therefore, it is of interest to compare $B_{\text{LU-CCRB}}(\thetavec_0,\Wmat)$ and $B_{\text{CCRB}}(\thetavec_0,\Wmat)$ in the asymptotic regime, {\it i.e.} when the number of independent identically distributed (i.i.d.) observation vectors tends to infinity. In the following proposition, we show the asymptotic relation between $B_{\text{LU-CCRB}}(\thetavec_0,\Wmat)$ and $B_{\text{CCRB}}(\thetavec_0,\Wmat)$.
\begin{proposition} \label{PROP_asymptotic}
Assume that $B_{\text{LU-CCRB}}(\thetavec_0,\Wmat)$ and $B_{\text{CCRB}}(\thetavec_0,\Wmat)$ exist and are nonzero. Then, given $L$ i.i.d. observation vectors,
\be\label{CCRB_is_higher_asymptitc_W}
\underset{L\to\infty}{\lim}\frac{B_{\text{LU-CCRB}}(\thetavec_0,\Wmat)}{B_{\text{CCRB}}(\thetavec_0,\Wmat)}=1.
\ee
\end{proposition}
\begin{proof}
For brevity, we remove the arguments of the functions that appear in this proof. Let $\Jmat^{(1)}$ denote the FIM based on a single observation vector. Then, for $L$ i.i.d. observation vectors (see e.g. \cite[p. 213]{KAY})
\be\label{J_N_asimptotic}
\Jmat=L\Jmat^{(1)}. 
\ee
By substituting \eqref{J_N_asimptotic} in \eqref{Gammamat_define}, we obtain
\be\label{Gammamat_define_asympt}
\Gammamat_{\Umat,\Wmat}=\Cmat_{\Umat,\Wmat}+L\left(\Umat^T\Wmat\Umat\right)\otimes\left(\Umat^T\Jmat^{(1)}\Umat\right).
\ee
Then, by substituting \eqref{Gammamat_define_asympt} in \eqref{total_CR_bound_CS}, we obtain the LU-CCRB on the WMSE based on $L$ i.i.d. observation vectors, which is given by
\be \label{Su_weighted_optimal_MSE_asympt}
\begin{split}
B_{\text{LU-CCRB}}&={\rm{vec}}^T\left(\Umat^T\Wmat\Umat\right)\left(\Cmat_{\Umat,\Wmat}+L\left(\Umat^T\Wmat\Umat\right)\right.\\
&~~~\left.\otimes\left(\Umat^T\Jmat^{(1)}\Umat\right)\right)^{\dagger}{\rm{vec}}\left(\Umat^T\Wmat\Umat\right).
\end{split}
\ee
By applying \eqref{kronecker_prop} on the right hand side (r.h.s.) of \eqref{W_CCRB_W} and using the properties of the trace and pseudo-inverse, it can be verified that
\be\label{Su_CCRB_W_new_expression}
\begin{split}
B_{\text{CCRB}}&={\rm{vec}}^T\left(\Umat^T\Wmat\Umat\right)\left(\left(\Umat^T\Wmat\Umat\right)^{\dagger}\right.\\
&~~~\left.\otimes\left(\Umat^T\Jmat\Umat\right)^{\dagger}\right){\rm{vec}}\left(\Umat^T\Wmat\Umat\right).
\end{split}
\ee
By using \eqref{kronecker_pseudo} and substituting \eqref{J_N_asimptotic} in \eqref{Su_CCRB_W_new_expression}, one obtains
\be\label{Su_CCRB_W_new_expression_new}
\begin{split}
B_{\text{CCRB}}&={\rm{vec}}^T\left(\Umat^T\Wmat\Umat\right)\left(L\left(\Umat^T\Wmat\Umat\right)\right.\\
&~~~\left.\otimes\left(\Umat^T\Jmat^{(1)}\Umat\right)\right)^{\dagger}{\rm{vec}}\left(\Umat^T\Wmat\Umat\right).
\end{split}
\ee
Under the assumption that $B_{\text{LU-CCRB}}$ exists, the elements of $\Cmat_{\Umat,\Wmat}$ are bounded while the term $L\left(\Umat^T\Wmat\Umat\right)\otimes\left(\Umat^T\Jmat^{(1)}\Umat\right)$ is proportional to $L$. Therefore, for $L\to\infty$ it can be verified from \eqref{Su_weighted_optimal_MSE_asympt} and \eqref{Su_CCRB_W_new_expression_new} that \eqref{CCRB_is_higher_asymptitc_W} is satisfied.	
\end{proof}
From Proposition \ref{PROP_asymptotic} it can be seen that $B_{\text{LU-CCRB}}(\thetavec_0,\Wmat)$ and $B_{\text{CCRB}}(\thetavec_0,\Wmat)$ asymptotically coincide. Consequently, under mild assumptions and similar to CCRB, $B_{\text{LU-CCRB}}(\thetavec_0,\Wmat)$ is asymptotically attained by the CML estimator. In addition, it should be noted that Proposition \ref{PROP_asymptotic} can be generalized and the term $\frac{B_{\text{LU-CCRB}}(\thetavecsmall_0,\Wmat)}{B_{\text{CCRB}}(\thetavecsmall_0,\Wmat)}$ tends to $1$ in any case where the FIM increases (in a matrix inequality sense), for example due to increasing signal-to-noise ratio, while the elements of $\Cmat_{\Umat,\Wmat}(\thetavec_0)$ are bounded.

\section{Examples} \label{sec:Examples}
In this section, we evaluate the proposed LU-CCRB in two scenarios. In the first scenario, we consider an orthogonal linear model with norm constraint and in the second scenario we consider complex amplitude estimation with amplitude constraint and unknown frequency. For both scenarios, it is shown that the CCRB on the WMSE is not a lower bound on the WMSE of the CML estimator in the non-asymptotic region. In contrast, we show that the CML estimator is a C-unbiased estimator and thus, the proposed LU-CCRB is a lower bound on the WMSE of the CML estimator. The CML estimator performance is computed using 10,000 Monte-Carlo trials.

\subsection{Linear model with norm constraint} \label{subsec:Linear model with norm constraint}
We consider the following linear observation model:
\be \label{model1}
\xvec=\Hmat\thetavec+\nvec,
\ee
where $\xvec\in{\mathbb{R}}^N$ is an observation vector, $\Hmat\in\mathbb{R}^{N\times M}$, $N\geq M$, is a known full-rank matrix, $\thetavec\in\mathbb{R}^{M}$ is an unknown deterministic parameter vector, and $\nvec\sim N(\zerovec,\sigma^2\Imat_N)$ is a zero-mean Gaussian noise vector with known covariance matrix $\sigma^2\Imat_N$. It is assumed that $\thetavec$ satisfies the norm constraint
\be\label{constraint_examp1}
f(\thetavec)=\|\thetavec\|^2-\rho^2=0,
\ee
where $\rho$ is known. This constraint arises, for example, in regularization techniques \cite{GOLUB,CHEN_REGULAR}. The CML estimator of $\thetavec$ satisfies
\be \label{CML_examp1}
\hat{\thetavec}_{\text{CML}}=\arg\underset{\thetavecsmall}{\min}~ \|\xvec-\Hmat\thetavec\|^2~~\text{s.t.}~~f(\thetavec)=0,
\ee
where $f(\thetavec)$ is given in \eqref{constraint_examp1}. By using \eqref{constraint_examp1}, we obtain $\Fmat(\thetavec)=2\thetavec^T$ and thus, $\Umat(\thetavec)$ satisfies 
\be\label{U_sphere_equality}
\Umat^T(\thetavec)\thetavec=\zerovec,~\Umat^T(\thetavec)\Umat(\thetavec)=\Imat_{M-1},~\forall\thetavec\in\Theta_\fvec,
\ee
where \eqref{U_sphere_equality} stems from \eqref{one}-\eqref{two}. In this example, we are interested in the trace of the MSE matrix and choose $\Wmat=\Imat_M$.\\
\indent
Under the model in \eqref{model1}, it can be shown that the FIM is given by
\be\label{FIM_examp1}
\Jmat(\thetavec)=\frac{1}{\sigma^2}\Hmat^T\Hmat.
\ee
By using \eqref{U_sphere_equality} and orthogonal projection matrix properties, we obtain
\be\label{P_u_dot_examp1}
\Pmat_\Umat^\bot(\thetavec)=\frac{1}{\rho^2}\thetavec\thetavec^T.
\ee
By substituting \eqref{P_u_dot_examp1} in \eqref{Gamma_block_define_I}, one obtains
\be\label{Gamma_block_define_I_examp1}
\Cmat_{\Umat,\Imat_M}^{(m,k)}(\thetavec)=\frac{1}{\rho^2}\Umat^T(\thetavec)\Vmat_m^T(\thetavec)\thetavec\thetavec^T\Vmat_k(\thetavec)\Umat(\thetavec),
\ee
$\forall m,k=1,\ldots,M-1$. It can be seen from \eqref{U_sphere_equality} that any column of $\Umat(\thetavec)$ satisfies
\be \label{u_m_sphere_equality}
\uvec_m^T(\thetavec)\thetavec=0,\forall \thetavec\in\mathbb{R}^M,~m=1,\ldots,M-1.
\ee
Taking the gradient of \eqref{u_m_sphere_equality} and using \eqref{Vmat_define}, we obtain
\be \label{u_m_sphere_equality_gradient}
\thetavec^T\Vmat_m(\thetavec)=-\uvec_m^T(\thetavec),\forall \thetavec\in\mathbb{R}^M,~m=1,\ldots,M-1.
\ee
By substituting \eqref{u_m_sphere_equality_gradient} in \eqref{Gamma_block_define_I_examp1}, one obtains
\be\label{Gamma_block_define_I_examp1_next}
\Cmat_{\Umat,\Imat_M}^{(m,k)}(\thetavec)=\frac{1}{\rho^2}\Umat^T(\thetavec)\uvec_m(\thetavec)\uvec_k^T(\thetavec)\Umat(\thetavec),
\ee
$\forall m,k=1,\ldots,M-1$. Then, by using $\Umat^T(\thetavec)\Umat(\thetavec)=\Imat_{M-1}$ and the block structure of $\Cmat_{\Umat,\Imat_M}(\thetavec)$, it can be verified that
\be\label{Gamma_define_I_examp1}
\Cmat_{\Umat,\Imat_M}(\thetavec)=\frac{1}{\rho^2}{\rm{vec}}\left(\Imat_{M-1}\right){\rm{vec}}^T\left(\Imat_{M-1}\right).
\ee
By substituting \eqref{FIM_examp1} and \eqref{Gamma_define_I_examp1} in \eqref{Gammamat_define_I}, we obtain
\be\label{Gammamat_define_I_examp1}
\begin{split}
\Gammamat_{\Umat,\Imat_M}(\thetavec)&=\frac{1}{\rho^2}{\rm{vec}}\left(\Imat_{M-1}\right){\rm{vec}}^T\left(\Imat_{M-1}\right)\\
&~~~+\frac{1}{\sigma^2}\Imat_{M-1}\otimes\left(\Umat^T(\thetavec)\Hmat^T\Hmat\Umat(\thetavec)\right).
\end{split}
\ee
By substituting \eqref{Gammamat_define_I_examp1} in \eqref{total_CR_bound_CS_I}, using Sherman--Morrison formula for matrix inversion (see e.g. \cite[p. 18]{MATRIX_COOKBOOK}) and the identity in \eqref{kronecker_prop}, and applying simple algebraic manipulations, one obtains
\be \label{total_CR_bound_CS_I_examp1_final}
B_{\text{LU-CCRB}}(\thetavec,\Imat_M)=\left(\frac{1}{\rho^2}+\frac{1}{{\rm{Tr}}\left(\Bmat_{\text{CCRB}}(\thetavec)\right)}\right)^{-1},
\ee
where the CCRB trace is given by
\be \label{W_CCRB_examp1}
{\rm{Tr}}\left(\Bmat_{\text{CCRB}}(\thetavec)\right)=\sigma^2{\rm{Tr}}\left(\left(\Umat^T(\thetavec)\Hmat^T\Hmat\Umat(\thetavec)\right)^{-1}\right),
\ee
as can be shown by substituting $\Wmat=\Imat_M$, \eqref{two}, and \eqref{FIM_examp1} in \eqref{W_CCRB_W}. It can be seen that the bound in \eqref{total_CR_bound_CS_I_examp1_final} is always lower than or equal to the bound in \eqref{W_CCRB_examp1}, where the gap between the bounds increases as $\rho$ decreases.

\subsubsection{Case 1 - $\Hmat$ has orthogonal columns with equal norms}\label{subsubsec:Case 1 - Hmat has orthogonal columns with equal norms}
We assume that $\Hmat$ satisfies
\be\label{H_equation}
\Hmat^T\Hmat=\beta\Imat_M,~\beta>0.
\ee
This assumption appears, for example, in \cite[pp. 88-90]{KAY}, \cite{Hero_constraint}, \cite{PROTTER}. By substituting \eqref{H_equation} in \eqref{W_CCRB_examp1}, we obtain the CCRB trace for this case
\be \label{W_CCRB_examp1_c1}
{\rm{Tr}}\left(\Bmat_{\text{CCRB}}(\thetavec)\right)=\frac{(M-1)\sigma^2}{\beta}.
\ee
The LU-CCRB for this case is obtained by substituting \eqref{W_CCRB_examp1_c1} in \eqref{total_CR_bound_CS_I_examp1_final} and is given by
\be \label{total_CR_bound_CS_I_examp1_final_c1}
B_{\text{LU-CCRB}}(\thetavec,\Imat_M)=\left(\frac{1}{\rho^2}+\frac{\beta}{(M-1)\sigma^2}\right)^{-1}.
\ee
It can be seen from \eqref{W_CCRB_examp1_c1} that the CCRB trace is not a function of the constrained norm value $\rho$. As opposed to the CCRB trace, the LU-CCRB from \eqref{total_CR_bound_CS_I_examp1_final_c1} is a function of $\rho$. By substituting \eqref{H_equation} in \eqref{CML_examp1}, it can be verified that the CML estimator for this case is given by
\be\label{CML_equation1}
\hat{\thetavec}_{\text{CML}}=\rho\frac{\Hmat^T\xvec}{\|\Hmat^T\xvec\|}.
\ee
This estimator is a uniformly C-unbiased estimator, as shown in the following proposition.
\begin{proposition} \label{PROP_C_unb_Sphere}
Under the model in \eqref{model1} and \eqref{H_equation}, the CML estimator from \eqref{CML_equation1} is a uniformly C-unbiased estimator for the weighting matrix $\Wmat=\Imat_M$. 
\end{proposition}
\begin{proof}
The proof is given in Appendix \ref{App_PROP_C_unb_Sphere}.
\end{proof}
\indent
In this example, we consider a single observation vector according to the model in \eqref{model1} and present simulations for $M=3$. In this case, it can be verified that the vectors $[\theta_2,-\theta_1,0]^T$, $[\theta_1\theta_3,\theta_2\theta_3,-\theta_1^2-\theta_2^2]^T$, and $\thetavec$ are mutually orthogonal. Thus, with proper normalization, we obtain
\be\label{U3_define}
\Umat(\thetavec)=\frac{1}{(\theta_1^2+\theta_2^2)^{\frac{1}{2}}}\begin{bmatrix}
\theta_2 & \frac{\theta_1\theta_3}{(\theta_1^2+\theta_2^2+\theta_3^2)^{\frac{1}{2}}}\\
-\theta_1 & \frac{\theta_2\theta_3}{(\theta_1^2+\theta_2^2+\theta_3^2)^{\frac{1}{2}}}\\
0 & \frac{-\theta_1^2-\theta_2^2}{(\theta_1^2+\theta_2^2+\theta_3^2)^{\frac{1}{2}}}
\end{bmatrix}
\ee
that satisfies \eqref{U_sphere_equality}. In addition, the vector $\thetavec$ can be written in spherical coordinates as
\begin{equation*}
\thetavec=[\rho\cos(\phi_1)\sin(\phi_2),\rho\sin(\phi_1)\sin(\phi_2),\rho\cos(\phi_2)]^T,
\end{equation*}
$\forall\thetavec\in\Theta_\fvec$, $\phi_1\in[-\pi,\pi),~\phi_2\in[0,\pi]$.\\
\indent 
In Fig. \ref{examp1_bias_phi_1_phi_2}, we examine the $\mathcal{X}$-unbiasedness of the CML estimator, as defined in \eqref{point_wise_chi_unb} and \eqref{local_wise_chi_unb}. The CML bias and the CML bias gradient multiplied by $\Umat(\thetavec)$ are evaluated numerically. The bias gradient is computed by using the equality $\Dmat_{\hat{\thetavecsmall}}(\thetavec_0)={\rm{E}}\left[(\hat{\thetavec}-\thetavec_0)\upsilonvec^T(\xvec,\thetavec_0)\right]-\Imat_M$ that stems from \eqref{log_expect} in Appendix \ref{App_T3}, where the expectation is computed via Monte-Carlo simulations. The bias of $\hat{\theta}_{\text{CML},1}$ and the term $[\Dmat_{\hat{\thetavecsmall}_{\text{CML}}}(\thetavec)\Umat(\thetavec)]_{1,1}$ are presented versus $\phi_1$ for $\phi_2=0.45\pi$ (upper) and versus $\phi_2$ for $\phi_1=0.2\pi$ (lower), where $\Hmat=\Imat_3$, $\sigma^2=16$, and $\rho=1$. It can be seen that the CML estimator, which is a uniformly C-unbiased estimator as proven in Proposition \ref{PROP_C_unb_Sphere}, does not satisfy \eqref{point_wise_chi_unb} and \eqref{local_wise_chi_unb}, {\it i.e.} it is not an $\mathcal{X}$-unbiased estimator.\\
\indent
In Fig. \ref{examp1_phi_1_phi_2}, the CCRB trace and the LU-CCRB are evaluated and compared to the MSE matrix trace of the CML estimator versus $\phi_1$ for $\phi_2=0.45\pi$ (upper) and versus $\phi_2$ for $\phi_1=0.2\pi$ (lower). It can be seen that $B_{\text{LU-CCRB}}(\thetavec,\Imat_M)$ is a lower bound on the MSE matrix trace of the CML estimator, while ${\rm{Tr}}\left(\Bmat_{\text{CCRB}}(\thetavec)\right)$ is not. The reason for this phenomenon is that the CML estimator is C-unbiased, as shown in Proposition \ref{PROP_C_unb_Sphere}, but it is not $\mathcal{X}$-unbiased as shown in Fig. \ref{examp1_bias_phi_1_phi_2}.\\
\indent
In Fig. \ref{examp1_R}, the CCRB trace and the LU-CCRB are evaluated versus $\rho$ and compared to the MSE matrix trace of the CML estimator for $\phi_1=0.2\pi$ and $\phi_2=0.45\pi$. It can be verified that the maximum likelihood (ML) estimator for this case, $\hat{\thetavec}_{\text{ML}}=\frac{\Hmat^T\xvec}{\beta}$, is a uniformly mean-unbiased estimator of $\thetavec$ but may not satisfy the constraint in \eqref{constraint_examp1} due to the additive noise in \eqref{model1}. The CML estimator from \eqref{CML_equation1} is obtained by normalizing this ML estimator, which forces it to satisfy the constraint. This normalization causes the CML estimator to be a mean-biased and an $\mathcal{X}$-biased estimator, but preserves its uniform C-unbiasedness, as proven in Proposition \ref{PROP_C_unb_Sphere}. Thus, in this case, C-unbiasedness is preserved under the normalization. It can be seen that the CCRB trace does not take the norm value into account even though the CML error is affected by the norm value, as manifested by the LU-CCRB. In addition, it can be seen that the CCRB trace is higher than the CML error for $\rho<9$. For sufficiently large values of $\rho$ ($\rho>40$), the LU-CCRB and the CCRB trace coincide, as can be observed from comparison of \eqref{W_CCRB_examp1_c1} and \eqref{total_CR_bound_CS_I_examp1_final_c1}.

\begin{figure}[h!]
\centering\includegraphics[width=7cm]{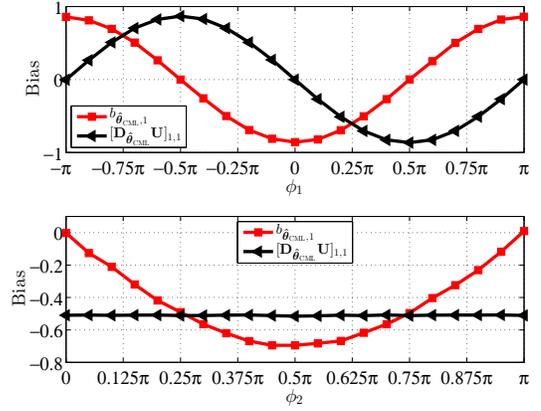}
\caption {Linear model with norm constraint, Case 1: The bias of $\hat{\theta}_{\text{CML},1}$, and the term $[\Dmat_{\hat{\thetavecsmall}_{\text{CML}}}(\thetavec)\Umat(\thetavec)]_{1,1}$ versus $\phi_1$ for $\phi_2=0.45\pi$ (top) and versus $\phi_2$ for $\phi_1=0.2\pi$ (bottom).
}\label{examp1_bias_phi_1_phi_2}
\end{figure}
\begin{figure}[h!]
\centering\includegraphics[width=7cm]{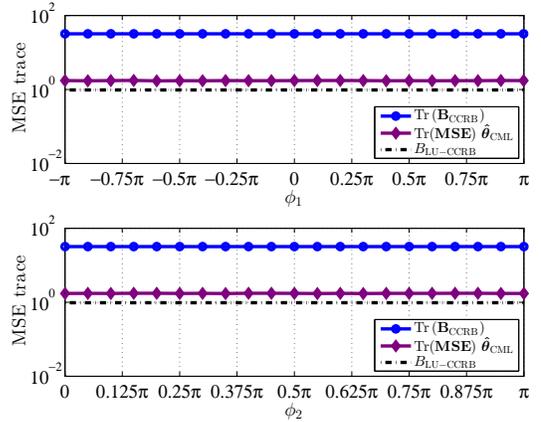}
\caption {Linear model with norm constraint, Case 1: MSE matrix trace of CML estimator, $B_{\text{LU-CCRB}}$, and ${\rm{Tr}}\left(\Bmat_{\text{CCRB}}\right)$ versus $\phi_1$ for $\phi_2=0.45\pi$ (top) and versus $\phi_2$ for $\phi_1=0.2\pi$ (bottom).
}\label{examp1_phi_1_phi_2}
\end{figure}
\begin{figure}[h!]
\centering\includegraphics[width=7cm]{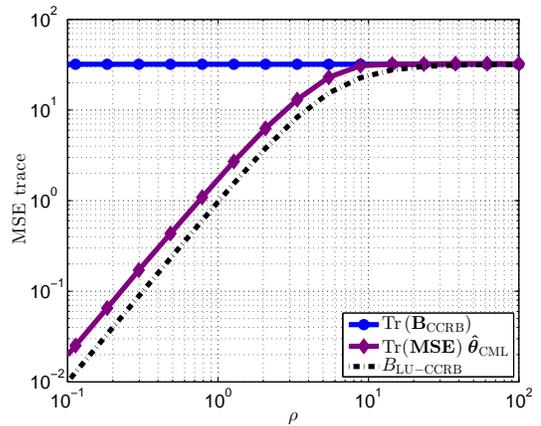}
\caption {Linear model with norm constraint, Case 1: MSE matrix trace of CML estimator, $B_{\text{LU-CCRB}}$, and ${\rm{Tr}}\left(\Bmat_{\text{CCRB}}\right)$ versus norm value $\rho$ for $\phi_1=0.2\pi$ and $\phi_2=0.45\pi$.
}\label{examp1_R}
\end{figure}

\subsubsection{Case 2 - General $\Hmat$}\label{subsubsec:Case 2 - General Hmat}
For the general case, {\it i.e.} when \eqref{H_equation} does not hold, the CML estimator from \eqref{CML_examp1} is evaluated numerically (see e.g. \cite{STOICA_LINEAR}, \cite[pp. 765-766]{MOON}, \cite{WIESEL_BECK}). In this case, we examine the asymptotic properties of the LU-CCRB and consider $L$ i.i.d. observation vectors following the model in \eqref{model1}, {\it i.e.}
\be \label{model1_2}
\xvec_l=\Hmat\thetavec+\nvec_l,~l=1,\ldots,L.
\ee
The model in \eqref{model1_2} can be presented as the single observation model in \eqref{model1} by replacing $\Hmat$ with
\be\label{H_case2}
\bar\Hmat=[\Hmat^T,\ldots,\Hmat^T]^T,
\ee
where $\bar\Hmat\in\mathbb{R}^{LN\times M}$. In the following simulation, we set $M=3$ and the LU-CCRB and CCRB trace for this case are obtained by substituting \eqref{U3_define} and \eqref{H_case2} in \eqref{total_CR_bound_CS_I_examp1_final} and \eqref{W_CCRB_examp1}, respectively. \\
\indent
In Fig. \ref{examp1_K}, the LU-CCRB and the CCRB trace are evaluated versus $L$ for $\Hmat=[\Imat_3,\hvec_4]^T$, $\hvec_4=[0.9,0.9,0.6]^T$, $\sigma^2=16$, $\phi_1=0.2\pi$, $\phi_2=0.45\pi$, and $\rho=1$. The considered bounds are compared to the MSE matrix trace of the CML estimator. It can be observed that although ${\rm{Tr}}\left(\Bmat_{\text{CCRB}}(\thetavec)\right)$ is asymptotically achieved by the CML estimator, it is greater than the MSE matrix trace of the CML estimator and does not provide a lower bound in the non-asymptotic region. In addition, it can be seen that for sufficiently large $L$, $B_{\text{LU-CCRB}}(\thetavec,\Imat_M)$ and ${\rm{Tr}}\left(\Bmat_{\text{CCRB}}(\thetavec)\right)$ coincide, which is in accordance with Proposition \ref{PROP_asymptotic}. Finally, it can be seen that $B_{\text{LU-CCRB}}(\thetavec,\Imat_M)$ and ${\rm{Tr}}\left(\Bmat_{\text{CCRB}}(\thetavec)\right)$ are attained by the CML estimator for $L>16$ and for $L>250$, respectively. Thus, this figure illustrates the distinction between the asymptotic region ($L>250$) and the non-asymptotic region ($L<250$) in terms of the relation between the CML estimator and the CCRB trace. In addition, this figure shows the inappropriateness of CCRB trace for predicting the behavior of the CML error.

\begin{figure}[h!]
\centering\includegraphics[width=7cm]{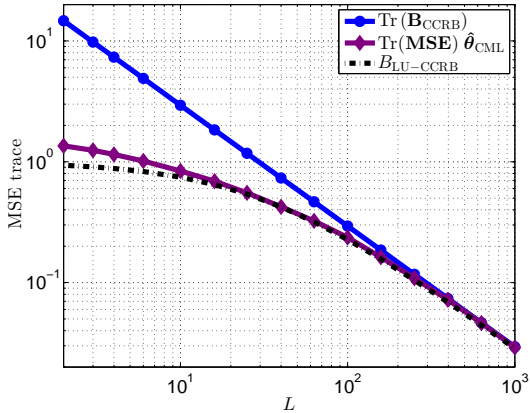}
\caption {Linear model with norm constraint, Case 2: MSE matrix trace of CML estimator, $B_{\text{LU-CCRB}}$, and ${\rm{Tr}}\left(\Bmat_{\text{CCRB}}\right)$ versus number of i.i.d. observations $L$.
}\label{examp1_K}
\end{figure}

\subsection{Complex amplitude estimation with amplitude constraint and unknown frequency} \label{subsec:Complex amplitude estimation with amplitude constraint and unknown frequency}
We consider complex amplitude estimation problem of a complex sinusoid embedded in noise, according to the following observation model \cite{RIFE,CYCLIC}:
\be \label{model2}
x_l=Ae^{jl\omega}+n_l,~l=l_1,l_1+1,\ldots,l_1+L-1,
\ee
where $l_1$ is the initial observation index, $L$ is the number of observations, the frequency, $\omega\in[-\pi,\pi)$, and the amplitude, $A\in{\mathbb{C}}$, are unknown, and $\left\{n_l\right\}_{l=l_1}^{l_1+L-1}$ is a complex circularly symmetric white Gaussian noise sequence with known variance $\sigma^2$. This example demonstrates a case in which $B_{\text{CCRB}}(\thetavec,\Wmat)$ is significantly affected by the choice of $l_1$ and diverges as $|l_1|\to\infty$, while the WMSE of the CML estimator is bounded. As opposed to $B_{\text{CCRB}}(\thetavec,\Wmat)$, it is shown that the corresponding LU-CCRB is a lower bound for the CML estimator and properly captures the CML error behavior. The unknown parameter vector for this model is $\thetavec=[\theta_1,\theta_2,\theta_3]^T=\left[{\text{Re}}\{A\},{\text{Im}}\{A\},\omega\right]^T$, where it is known that the amplitude $A$ satisfies the constraint 
\be\label{amplitude_constraint}
|A|^2={\text{Re}}^2\{A\}+{\text{Im}}^2\{A\}=c^2
\ee
or equivalently, $f(\thetavec)=\theta_1^2+\theta_2^2-c^2=0$. This constraint can be viewed as a known radar budget in a Doppler estimation problem \cite{CCRB_hybrid}. It should be noted that we can reparameterize this problem and remove the constraint by directly estimating $\angle{A}$. However, there is no uniformly mean-unbiased estimator of this phase \cite{TODROS_WINNIK}, \cite{PHASE_KAY}. In this case, $\Fmat(\thetavec)=2[\theta_1,\theta_2,0]$ and it can be verified that the vectors $[\theta_2,-\theta_1,0]^T$, $[0,0,1]^T$, and $2[\theta_1,\theta_2,0]^T$ are mutually orthogonal. Thus, with proper normalization, we obtain
\be\label{U_freq_define}
\Umat(\thetavec)=\begin{bmatrix}
\frac{\theta_2}{(\theta_1^2+\theta_2^2)^{\frac{1}{2}}} & 0  \\
-\frac{\theta_1}{(\theta_1^2+\theta_2^2)^{\frac{1}{2}}} & 0 \\
0 & 1
\end{bmatrix}
\ee
that satisfies \eqref{one}-\eqref{two}. The CML estimator of $\thetavec$ under the constraint in \eqref{amplitude_constraint} is given by \cite{RIFE}
\begin{equation}\label{CML_freq}
\hat\thetavec_{\text{CML}}=\left[\frac{c{\text{Re}}\{Y(\xvec,\hat{\omega}_{\text{CML}})\}}{|Y(\xvec,\hat{\omega}_{\text{CML}})|},\frac{c{\text{Im}}\{Y(\xvec,\hat{\omega}_{\text{CML}})\}}{|Y(\xvec,\hat{\omega}_{\text{CML}})|},\hat{\omega}_{\text{CML}}\right]^T,
\end{equation}
where
$\hat{\omega}_{\text{CML}}=\arg\underset{\omega\in[-\pi,\pi)}{\max}\left|Y(\xvec,\omega)\right|^2$ and $Y(\xvec,\omega)\define\frac{1}{L}\sum_{l=l_1}^{l_1+L-1}x_l e^{-jl\omega}$. In this example, we focus on the estimation of the amplitude that is affected by the constraint in \eqref{amplitude_constraint}. Thus, we evaluate the WMSE from \eqref{WMSE} with weighting matrix
\be\label{W_I2}
\Wmat=\begin{bmatrix}
1 & 0 & 0\\
0 & 1 & 0\\
0 & 0 & 0
\end{bmatrix}
\ee
and the unknown frequency can be considered as a nuisance parameter that may affect the amplitude estimation performance. Let $\Umat^T(\thetavec)\Wmat\bvec_{\hat{\thetavecsmall}}(\thetavec)$ denote the C-bias of an estimator. Since analytic computation of the CML C-bias is intractable in this case, we evaluate its C-bias norm numerically. Similarly, we examine the $\mathcal{X}$-unbiasedness of the CML estimator by numerical evaluation of the CML bias and the CML bias gradient multiplied by $\Umat(\thetavec)$, as explained in Subsection \ref{subsubsec:Case 1 - Hmat has orthogonal columns with equal norms}.\\
\indent
The FIM in this case is given by
\be\label{FIM_examp2}
\Jmat(\thetavec)=\frac{2L}{\sigma^2}\begin{bmatrix}
1 & 0 & -\frac{\theta_2 (2l_1+L-1)}{2}\\
0 & 1 & \frac{\theta_1 (2l_1+L-1)}{2}\\
-\frac{\theta_2 (2l_1+L-1)}{2} & \frac{\theta_1 (2l_1+L-1)}{2} & \frac{c^2 \sum_{l=l_1}^{l_1+L-1}l^2}{L}
\end{bmatrix},
\ee
$\forall\thetavec\in\Theta_\fvec$. By substituting \eqref{U_freq_define}, \eqref{W_I2}, and \eqref{FIM_examp2} in \eqref{W_CCRB_W}, we obtain
\be \label{CCRB_examp2}
B_{\text{CCRB}}(\thetavec,\Wmat)=\sigma^2\frac{6l_1^2+6(L-1)l_1+(2L-1)(L-1)}{L(L-1)(L+1)}.
\ee
It is shown in Appendix \ref{App_LU_CCRB_examp2} that the LU-CCRB from \eqref{total_CR_bound_CS} for this case, is given by
\be \label{LU_CCRB_examp2}
B_{\text{LU-CCRB}}(\thetavec,\Wmat)=\left(\frac{1}{c^2}+\frac{1}{B_{\text{CCRB}}(\thetavec,\Wmat)}\right)^{-1}.
\ee
It can be seen that the CCRB in \eqref{CCRB_examp2} is not a function of the constraint parameter, $c$, and therefore, the insight it provides on the system may be inadequate.\\
\indent
In the upper plot of Fig. \ref{examp2_phi_bias_mse}, the C-bias norm of $\hat\thetavec_{\text{CML}}$, the bias of $\hat{\theta}_{\text{CML},1}$, and the term $[\Dmat_{\hat{\thetavecsmall}_{\text{CML}}}(\thetavec)\Umat(\thetavec)]_{1,1}$ are presented versus $\angle{A}$ for $c=0.2$, where $l_1=1$, $L=15$, $\omega=0.9\pi$, and $\sigma^2=16$. It can be seen that the C-bias norm of $\hat\thetavec_{\text{CML}}$ is approximately zero. In addition, we observed in the simulations that in this scenario the CML C-bias norm is approximately zero in the {\em entire} parameter space, {\it i.e.} $\forall\thetavec\in\Theta_\fvec$. Thus, $\hat\thetavec_{\text{CML}}$ can be considered as a C-unbiased estimator. In addition, it can be observed that $\hat\thetavec_{\text{CML}}$ is not an $\mathcal{X}$-unbiased estimator, since it does not satisfy \eqref{point_wise_chi_unb} and \eqref{local_wise_chi_unb}. The LU-CCRB and the CCRB from \eqref{LU_CCRB_examp2} and \eqref{CCRB_examp2}, respectively, are evaluated in the lower plot of Fig. \ref{examp2_phi_bias_mse} versus $\angle{A}$ and compared to the WMSE of the CML estimator. It can be seen that the CML estimator has uniformly lower WMSE than the corresponding CCRB, while the LU-CCRB is a lower bound on its WMSE.\\
\indent
For a fixed number of observations, $L$, it can be seen from \eqref{CCRB_examp2} that for $|l_1|\to\infty$ the CCRB tends to infinity, while the LU-CCRB does not. By using \eqref{CML_freq} and \eqref{W_I2}, it can be verified that for any $l_1\in\mathbb{Z}$ the WMSE of the CML estimator is bounded. In Fig. \ref{examp2_l1}, the LU-CCRB from \eqref{LU_CCRB_examp2} and the CCRB from \eqref{CCRB_examp2} are evaluated versus $l_1$ and compared to the WMSE of the CML estimator, where $\angle{A}=0.3\pi$. It can be seen that $B_{\text{CCRB}}(\thetavec,\Wmat)$ is not a lower bound for the CML performance and that the behavior of this bound is misleading, as $|l_1|\to\infty$, since estimators with finite WMSE for any $l_1\in\mathbb{Z}$ can be found. As opposed to $B_{\text{CCRB}}(\thetavec,\Wmat)$, the LU-CCRB is a lower bound on the WMSE of the CML estimator in this case.\\
\indent
In Fig. \ref{examp2_SNR}, the LU-CCRB and the corresponding CCRB from \eqref{LU_CCRB_examp2} and \eqref{CCRB_examp2}, respectively, are evaluated versus $\frac{1}{\sigma^2}$ and compared to the WMSE of the CML estimator for $c=0.2$ and $c=0.5$, where $\angle{A}=0.3\pi$. In this case, the FIM from \eqref{FIM_examp2} increases (in a matrix inequality sense) as $\frac{1}{\sigma^2}$ increases and the elements of $\Cmat_{\Umat,\Wmat}(\thetavec)$ from \eqref{Gamma_block_define_examp2} in Appendix \ref{App_LU_CCRB_examp2} are bounded. Thus, it can be seen that $B_{\text{LU-CCRB}}(\thetavec,\Wmat)$ and $B_{\text{CCRB}}(\thetavec,\Wmat)$ coincide for sufficiently large values of $\frac{1}{\sigma^2}$, as explained in Subsection \ref{subsubsec:Asymptotic properties}. However, it can be seen that unlike $B_{\text{CCRB}}(\thetavec,\Wmat)$, the proposed $B_{\text{LU-CCRB}}(\thetavec,\Wmat)$ provides a lower bound on the WMSE of the CML estimator for any value of $\frac{1}{\sigma^2}$. In addition, it can be seen that both the LU-CCRB and the CML error are affected by the constraint parameter, $c$, which is not taken into account by the CCRB from \eqref{CCRB_examp2}. Finally, this figure shows that even though $B_{\text{CCRB}}(\thetavec,\Wmat)$ is attained by the CML estimator for sufficiently high values of $\frac{1}{\sigma^2}$ ($\frac{1}{\sigma^2}>10^{0.75}$ for $c=0.5$ and $\frac{1}{\sigma^2}>10^{1.5}$ for $c=0.2$), it is not informative and is not appropriate for performance analysis of the CML estimator in lower values of $\frac{1}{\sigma^2}$, while the LU-CCRB is informative for any value of $\frac{1}{\sigma^2}$.\\
\indent
In Fig. \ref{examp2_N}, $B_{\text{LU-CCRB}}(\thetavec,\Wmat)$ and $B_{\text{CCRB}}(\thetavec,\Wmat)$ from \eqref{LU_CCRB_examp2} and \eqref{CCRB_examp2}, respectively, are evaluated versus the number of observations $L$ and compared to the WMSE of the CML estimator for $l_1=1$, $c=1$, and $\sigma^2=16$. It can be seen that the LU-CCRB is an informative lower bound on the WMSE of the CML estimator for any value of $L$, while $B_{\text{CCRB}}(\thetavec,\Wmat)$ is not a lower bound for $L<20$. In addition, it can be seen that for $L>450$, the bounds coincide and are attained by the CML estimator.

\begin{figure}[h!]
\centering\includegraphics[width=7cm]{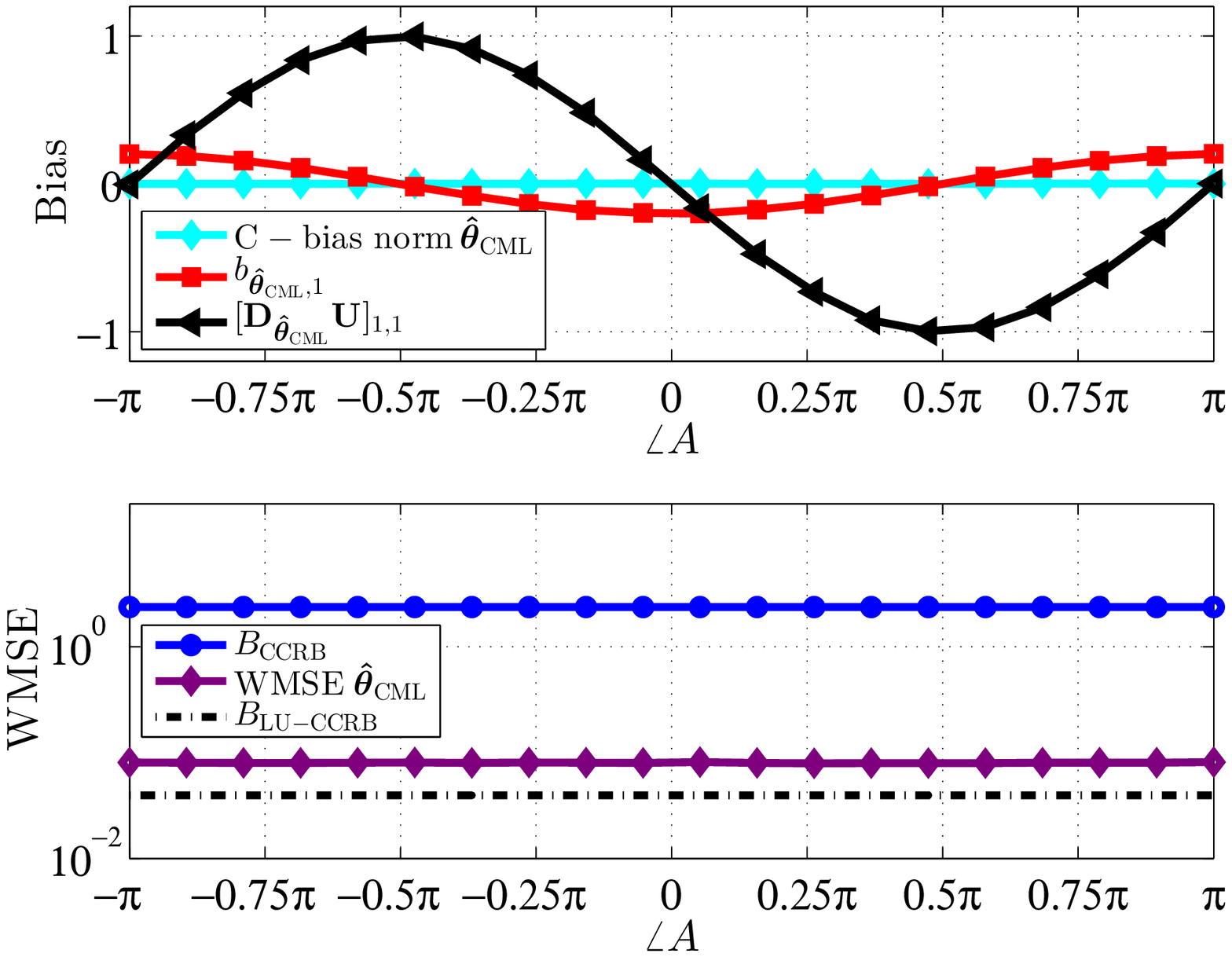}
\caption {Complex amplitude estimation with amplitude constraint and unknown frequency: The C-bias norm of $\hat{\thetavec}_{\text{CML}}$, the bias of $\hat{\theta}_{\text{CML},1}$, and the term $[\Dmat_{\hat{\thetavecsmall}_{\text{CML}}}(\thetavec)\Umat(\thetavec)]_{1,1}$ versus $\angle{A}$ (top). WMSE of CML estimator, $B_{\text{LU-CCRB}}$, and $B_{\text{CCRB}}$ versus $\angle{A}$ (bottom).
}\label{examp2_phi_bias_mse}
\end{figure}

\begin{figure}[h!]
\centering\includegraphics[width=7cm]{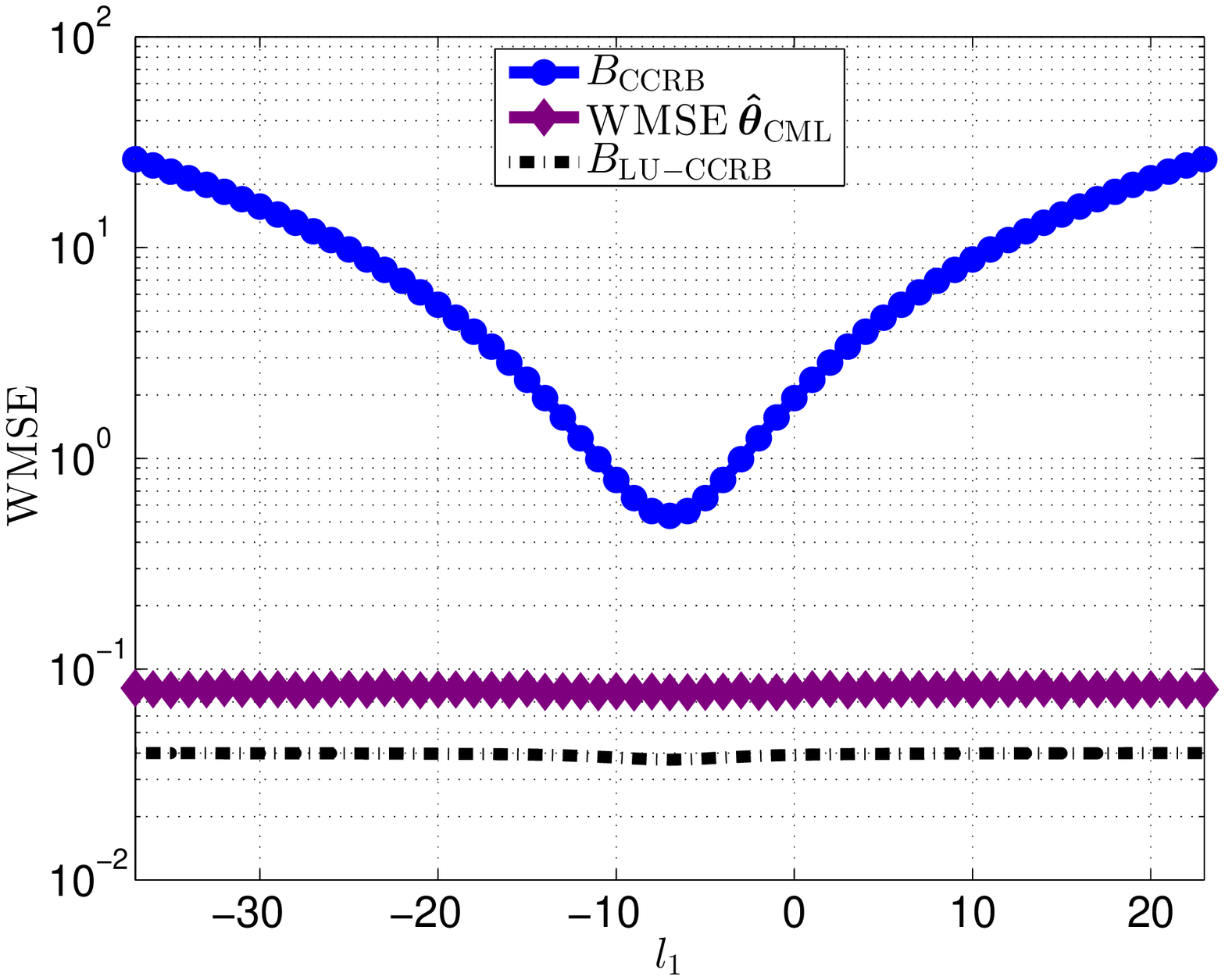}
\caption {Complex amplitude estimation with amplitude constraint and unknown frequency: WMSE of CML estimator, $B_{\text{LU-CCRB}}$, and $B_{\text{CCRB}}$ versus $l_1$.
}\label{examp2_l1}
\end{figure}

\begin{figure}[h!]
\centering\includegraphics[width=7cm]{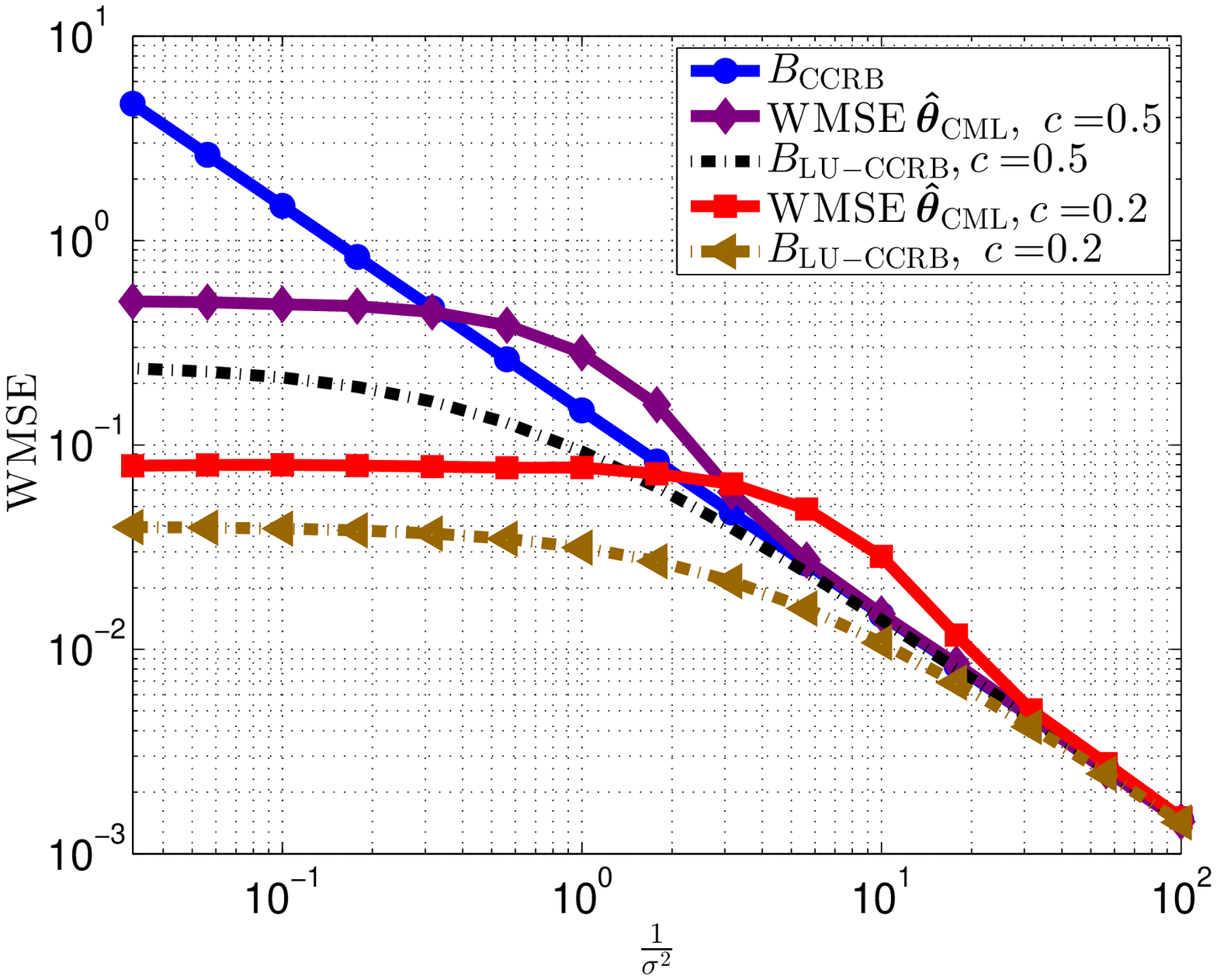}
\caption {Complex amplitude estimation with amplitude constraint and unknown frequency: WMSE of CML estimator, $B_{\text{LU-CCRB}}$, and $B_{\text{CCRB}}$ versus $\frac{1}{\sigma^2}$.
}\label{examp2_SNR}
\end{figure}

\begin{figure}[h!]
\centering\includegraphics[width=7cm]{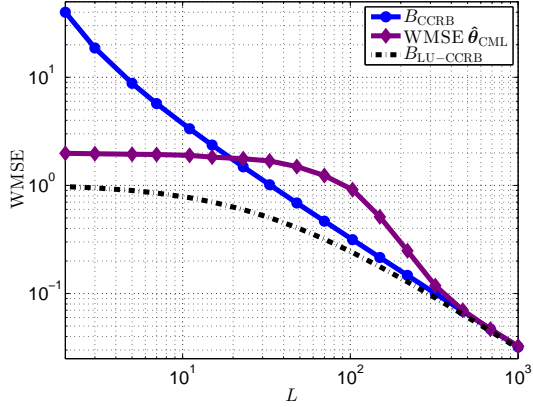}
\caption {Complex amplitude estimation with amplitude constraint and unknown frequency: WMSE of CML estimator, $B_{\text{LU-CCRB}}$, and $B_{\text{CCRB}}$ versus $L$.
}\label{examp2_N}
\end{figure}

\section{Conclusion} \label{sec:Conclusion}
In this paper, we consider non-Bayesian parameter estimation under parametric constraints. First, a novel unbiasedness restriction, denoted by C-unbiasedness, is proposed. The C-unbiasedness is based on Lehmann's concept of unbiasedness that takes into account the chosen cost function and the relevant parameter space. In addition, we propose a novel Cram$\acute{\text{e}}$r-Rao-type bound, denoted by LU-CCRB, which is a lower bound on the WMSE of locally C-unbiased estimators. The properties of LU-CCRB and its relation to the well-known CCRB were examined both analytically and via simulations. Simulations with two examples show that the CML estimator is C-unbiased but does not satisfy the restrictive CCRB unbiasedness conditions. As a result, in these examples the corresponding CCRB is not a lower bound on the WMSE of the CML estimator in the non-asymptotic region, while the proposed LU-CCRB is an informative lower bound.

\appendices

\section{Proof of Proposition \ref{LocCunbias_prop}}\label{App_LocCunbias_prop}
Let 
\be \label{gvec_define_loc}
\gvec(\thetavec)\define\Umat^T(\thetavec)\Wmat\bvec_{\hat{\thetavecsmall}}(\thetavec).
\ee
Based on the proof of Proposition \ref{Cunbias_prop} and using Definition \ref{Def_local_Cunb} with $\Omega_\thetavecsmall=\Theta_\fvec$, it can be verified that a necessary condition for local Lehmann-unbiasedness in the vicinity of $\thetavec_0\in\Theta_\fvec$ w.r.t. the WSE cost function is
\be \label{Loc_defdef2_WSE_g}
\gvec(\thetavec)=\zerovec,
\ee
for any $\thetavec\in\Theta_\fvec$, s.t. $|\theta_m-\theta_{0,m}|<\varepsilon_m,~\varepsilon_m\to0,~\forall m=1,\ldots,M$. In the vicinity of $\thetavec_0$, by using the notion of feasible directions under equality constraints \cite{BenHaim,sparse_con,normC} and using \eqref{one}-\eqref{two}, we can write
\be\label{theta_lims}
\thetavec=\thetavec_0+\Umat(\thetavec_0)\dvec\tauvec,~\dvec\tauvec\in{\mathbb{R}}^{M-K},~\|\dvec\tauvec\|<\varepsilon,~\varepsilon\to0.
\ee
Then, by substituting \eqref{theta_lims} in the left hand side (l.h.s.) of \eqref{Loc_defdef2_WSE_g} and rewriting in terms of Taylor expansion, one obtains
\be \label{sec_localU_taylor_g}
\gvec(\thetavec_0)+\left.\nabla_{\thetavecsmall}\gvec(\thetavec)\right|_{\thetavecsmall=\thetavecsmall_0}\Umat(\thetavec_0)\dvec\tauvec+{\bf{o}}(\|\dvec\tauvec\|)=\zerovec,
\ee
$\dvec\tauvec\in{\mathbb{R}}^{M-K},~\|\dvec\tauvec\|<\varepsilon,~\varepsilon\to0$. From the uniqueness of Taylor expansion, one obtains the equalities
\be \label{point_wise_unb_g}
\gvec(\thetavec_0)=\zerovec
\ee
and
\be \label{local_wise_unb_g}
\left.\nabla_{\thetavecsmall}\gvec(\thetavec)\right|_{\thetavecsmall=\thetavecsmall_0}\Umat(\thetavec_0)=\zerovec.
\ee
By substituting \eqref{gvec_define_loc} in \eqref{point_wise_unb_g} and \eqref{local_wise_unb_g}, we obtain the first local C-unbiasedness condition from \eqref{point_wise_unb} and
\be \label{sec_localU_final_temp}
\left.\nabla_{\thetavecsmall}\left(\Umat^T(\thetavec)\Wmat\bvec_{\hat{\thetavecsmall}}(\thetavec)\right)\right|_{\thetavecsmall=\thetavecsmall_0}\Umat(\thetavec_0)=\zerovec,
\ee
respectively. The equality in \eqref{sec_localU_final_temp} can be rewritten in terms of the following $M-K$ equalities
\be \label{sec_localU_final_temp_M_minus_K}
\left.\nabla_{\thetavecsmall}\left(\uvec_m^T(\thetavec)\Wmat\bvec_{\hat{\thetavecsmall}}(\thetavec)\right)\right|_{\thetavecsmall=\thetavecsmall_0}\Umat(\thetavec_0)=\zerovec,
\ee
$\forall m=1,\ldots,M-K$. By using the product rule for derivatives we can rewrite \eqref{sec_localU_final_temp_M_minus_K} as
\begin{multline} \label{sec_localU_final_temp_M_minus_K_next}
\bvec_{\hat{\thetavecsmall}}^T(\thetavec_0)\Wmat\left(\left.\nabla_{\thetavecsmall}\uvec_m(\thetavec)\right|_{\thetavecsmall_0}\right)\Umat(\thetavec_0)\\
+\uvec_m^T(\thetavec_0)\Wmat\left(\left.\nabla_{\thetavecsmall}\bvec_{\hat{\thetavecsmall}}(\thetavec)\right|_{\thetavecsmall_0}\right)\Umat(\thetavec_0)=\zerovec,
\end{multline}
$\forall m=1,\ldots,M-K$. Finally, by substituting \eqref{Vmat_define} and \eqref{bias_gradient} in \eqref{sec_localU_final_temp_M_minus_K_next} and reordering, we obtain the second local C-unbiasedness condition from \eqref{sec_localU_final}.

\section{Proof of Theorem \ref{T3}}\label{App_T3}
Under Conditions \ref{cond1CRB}-\ref{cond3CRB}, by using Cauchy–-Schwarz inequality, it can be verified that
\begin{multline}\label{Su_substitution_CS}
{\rm{E}}\left[(\hat{\thetavec}-\thetavec_0)^T\Wmat(\hat{\thetavec}-\thetavec_0)\right]{\rm{E}}\left[\xivec_{\Wmat}^T(\xvec,\thetavec_0)\Wmat\xivec_{\Wmat}(\xvec,\thetavec_0)\right]\\
\geq\left({\rm{E}}\left[\xivec_{\Wmat}^T(\xvec,\thetavec_0)\Wmat(\hat{\thetavec}-\thetavec_0)\right]\right)^2
\end{multline}
for any positive semidefinite weighting matrix, $\Wmat$, where
\be\label{Su_b_aux_W}
\xivec_{\Wmat}(\xvec,\thetavec_0)\define\sum_{m=1}^{M-K}\left(\Smat^{(m)}_{\Wmat}(\thetavec_0)+\right.\left.\Tmat^{(m)}_{\Wmat}(\xvec,\thetavec_0)\right)\cvec_{m}
\ee
is an auxiliary function, $\Smat^{(m)}_{\Wmat}(\thetavec_0)$ and $\Tmat^{(m)}_{\Wmat}(\xvec,\thetavec_0)$ are defined in \eqref{S_define} and \eqref{T_define}, respectively, and $\cvec_{m}\in\mathbb{R}^{M-K}$ is an arbitrary vector, $\forall m=1,\ldots,M-K$. This auxiliary function allows the derivation of an estimator-independent lower bound.\footnote{It is shown in Appendix \ref{App_Alternative derivation of LU-CCRB} that the LU-CCRB from \eqref{total_CR_bound_CS_I} can also be derived by minimization of the MSE matrix trace under the local C-unbiasedness constraints stated in \eqref{point_wise_unb_I}-\eqref{sec_localU_final_I}. This alternative derivation elucidates the choice of the auxiliary function from \eqref{Su_b_aux_W} in the proof of Theorem \ref{T3}.} By using \eqref{Su_b_aux_W}, we obtain
\be\label{Su_DENOM_CS}
\begin{split}
&{\rm{E}}\left[\xivec_{\Wmat}^T(\xvec,\thetavec_0)\Wmat\xivec_{\Wmat}(\xvec,\thetavec_0)\right]\\
&=\sum_{m=1}^{M-K}\sum_{k=1}^{M-K}\cvec_{m}^T{\rm{E}}\left[\left(\Smat^{(m)}_{\Wmat}(\thetavec_0)+\Tmat^{(m)}_{\Wmat}(\xvec,\thetavec_0)\right)^T\right.\\
&~~~\left.\times\Wmat\left(\Smat^{(k)}_{\Wmat}(\thetavec_0)+\Tmat^{(k)}_{\Wmat}(\xvec,\thetavec_0)\right)\right]\cvec_{k}.
\end{split}
\ee 
Under regularity condition \ref{cond1CRB}, it can be shown that \cite[p. 67]{KAY} 
\be\label{smoothness}
{\rm{E}}\left[\upsilonvec(\xvec,\thetavec_0)\right]=\zerovec,
\ee
where $\upsilonvec(\xvec,\thetavec)$ is the log-likelihood derivative, defined in \eqref{l_define}. By using \eqref{smoothness} and \eqref{T_define}, we obtain
\be\label{Su_E_T_ZERO_W}
{\rm{E}}\left[\Tmat^{(m)}_{\Wmat}(\xvec,\thetavec_0)\right]=\zerovec,~\forall m=1,\ldots,M-K. 
\ee
Thus,
\be\label{Su_firstStepAppA}
\begin{split} 
&{\rm{E}}\left[\left(\Smat^{(m)}_{\Wmat}(\thetavec_0)+\Tmat^{(m)}_{\Wmat}(\xvec,\thetavec_0)\right)^T\right.\\
&~~~\left.\times\Wmat\left(\Smat^{(k)}_{\Wmat}(\thetavec_0)+\Tmat^{(k)}_{\Wmat}(\xvec,\thetavec_0)\right)\right]\\
&=\left(\Smat^{(m)}_{\Wmat}(\thetavec_0)\right)^T\Wmat\Smat^{(k)}_{\Wmat}(\thetavec_0)\\
&~~~+{\rm{E}}\left[\left(\Tmat^{(m)}_{\Wmat}(\xvec,\thetavec_0)\right)^T\Wmat\Tmat^{(k)}_{\Wmat}(\xvec,\thetavec_0)\right]\\
&=\Cmat_{\Umat,\Wmat}^{(m,k)}(\thetavec_0)+\left(\uvec_m^T(\thetavec_0)\Wmat\uvec_k(\thetavec_0)\right)\Umat^T(\thetavec_0)\Jmat(\thetavec_0)\Umat(\thetavec_0),
\end{split}
\ee
$\forall m,k=1,\ldots,M-K$, where the first equality is obtained by substituting \eqref{Su_E_T_ZERO_W}. The second equality is obtained by substituting \eqref{Gamma_block_define} and \eqref{T_define}, using pseudo-inverse properties, applying simple algebraic manipulations, and substituting \eqref{FIM}. Then, by substituting \eqref{Su_firstStepAppA} in \eqref{Su_DENOM_CS} and using Kronecker product definition and \eqref{Gammamat_define}, we obtain 
\be\label{Su_DENOM_CS_next}
{\rm{E}}\left[\xivec_{\Wmat}^T(\xvec,\thetavec_0)\Wmat\xivec_{\Wmat}(\xvec,\thetavec_0)\right]=\cvec^T\Gammamat_{\Umat,\Wmat}(\thetavec_0)\cvec,
\ee
where
\be\label{Su_cvec_define_W}
\cvec\define[\cvec_1^T,\ldots,\cvec_{M-K}^T]^T. 
\ee
Next, by using \eqref{Su_b_aux_W} and \eqref{bias_definition}, one obtains
\begin{multline}\label{Su_numer_CS_tir_next}
{\rm{E}}\left[\xivec_{\Wmat}^T(\xvec,\thetavec_0)\Wmat(\hat{\thetavec}-\thetavec_0)\right]=\sum_{m=1}^{M-K}\bigg(\cvec_m^T\left(\Smat^{(m)}_{\Wmat}(\thetavec_0)\right)^T\\
\times\Wmat\bvec_{\hat{\thetavecsmall}}(\thetavec_0)\bigg)+\sum_{m=1}^{M-K}\cvec_m^T{\rm{E}}\left[\left(\Tmat^{(m)}_{\Wmat}(\xvec,\thetavec_0)\right)^T\Wmat(\hat{\thetavec}-\thetavec_0)\right].
\end{multline}
By substituting \eqref{S_define} and \eqref{T_define} in \eqref{Su_numer_CS_tir_next} and using pseudo-inverse and trace properties, we obtain
\be\label{Su_numer_CS_tir_next4}
\begin{split}
&{\rm{E}}\left[\xivec_{\Wmat}^T(\xvec,\thetavec_0)\Wmat(\hat{\thetavec}-\thetavec_0)\right]\\
&=\sum_{m=1}^{M-K}\cvec_m^T\Umat^T(\thetavec_0)\Vmat_m^T(\thetavec_0)\Wmat^{\frac{1}{2}}\\
&~~~\times\Pmat_{\Wmat^{\frac{1}{2}}\Umat}^\bot(\thetavec_0)\Wmat^{\frac{1}{2}}\Wmat^{\dagger}\Wmat\bvec_{\hat{\thetavecsmall}}(\thetavec_0)\\
&+\sum_{m=1}^{M-K}{\rm{Tr}}\left({\rm{E}}\left[(\hat{\thetavec}-\thetavec_0)\upsilonvec^T(\xvec,\thetavec_0)\right]\Umat(\thetavec_0)\cvec_m\uvec_m^T(\thetavec_0)\Wmat\right),
\end{split}
\ee
By using the definition of orthogonal projection matrix and pseudo-inverse properties, it can be verified that for any locally C-unbiased estimator
\be\label{Su_numer_express_II}
\begin{split}
&\Wmat^{\frac{1}{2}}\Pmat_{\Wmat^{\frac{1}{2}}\Umat}^\bot(\thetavec_0)\Wmat^{\frac{1}{2}}\Wmat^{\dagger}\Wmat\bvec_{\hat{\thetavecsmall}}(\thetavec_0)\\
&=\left(\Wmat-\Wmat\Umat(\thetavec_0)(\Umat^T(\thetavec_0)\Wmat\Umat(\thetavec_0))^{\dagger}\Umat^T(\thetavec_0)\Wmat\right)\bvec_{\hat{\thetavecsmall}}(\thetavec_0)\\
&=\Wmat\bvec_{\hat{\thetavecsmall}}(\thetavec_0).
\end{split}
\ee
where the second equality is obtained by substituting the first condition for local C-unbiasedness from \eqref{point_wise_unb}. From Condition \ref{cond1CRB}, it can be verified that \cite{BenHaim}
\be \label{log_expect}
{\rm{E}}\left[(\hat{\thetavec}-\thetavec_0)\upsilonvec^T(\xvec,\thetavec_0)\right]=\Imat_M+\Dmat_{\hat{\thetavecsmall}}(\thetavec_0).
\ee
By substituting \eqref{log_expect} and \eqref{Su_numer_express_II} in \eqref{Su_numer_CS_tir_next4} and using trace properties and some algebraic manipulations, one obtains
\be\label{Su_numer_CS_tir_next5}
\begin{split}
&{\rm{E}}\left[\xivec_{\Wmat}^T(\xvec,\thetavec_0)\Wmat(\hat{\thetavec}-\thetavec_0)\right]\\
&=\sum_{m=1}^{M-K}\left(\bvec_{\hat{\thetavecsmall}}^T(\thetavec_0)\Wmat\Vmat_m(\thetavec_0)\Umat(\thetavec_0)+\uvec_m^T(\thetavec_0)\Wmat\Umat(\thetavec_0)\right.\\
&~~~\left.+\uvec_m^T(\thetavec_0)\Wmat\Dmat_{\hat{\thetavecsmall}}(\thetavec_0)\Umat(\thetavec_0)\right)\cvec_m=\psivec_{\Wmat}^T(\thetavec_0)\cvec,
\end{split}
\ee
where the second equality is obtained by substituting the second condition for local C-unbiasedness from \eqref{sec_localU_final}, \eqref{Su_cvec_define_W}, and 
\be\label{evec_define}
\psivec_{\Wmat}(\thetavec)\define{\rm{vec}}\left(\Umat^T(\thetavec)\Wmat\Umat(\thetavec)\right)\in\mathbb{R}^{(M-K)^2}.
\ee
By substituting \eqref{Su_DENOM_CS_next} and \eqref{Su_numer_CS_tir_next5} in \eqref{Su_substitution_CS}, one obtains
\be\label{Su_substitution_CS_newer}
{\rm{E}}\left[(\hat{\thetavec}-\thetavec_0)^T\Wmat(\hat{\thetavec}-\thetavec_0)\right]\cvec^T\Gammamat_{\Umat,\Wmat}(\thetavec_0)\cvec\geq\left(\psivec_{\Wmat}^T(\thetavec_0)\cvec\right)^2.
\ee
Finally, by substituting 
\be\label{Su_c_opt}
\cvec=\Gammamat_{\Umat,\Wmat}^{\dagger}(\thetavec_0)\psivec_{\Wmat}(\thetavec_0)
\ee
in \eqref{Su_substitution_CS_newer}, reordering, using pseudo-inverse matrix property \cite[p. 21]{MATRIX_COOKBOOK}
\begin{equation}\label{pinv_A_dagger}
\Amat^{\dagger}\Amat\Amat^{\dagger}=\Amat^{\dagger},
\end{equation}
and substituting \eqref{evec_define}, we obtain the LU-CCRB from \eqref{total_CR_bound_CS}. It is shown in Appendix \ref{App_Derivation of Su_c_opt} that the choice of $\cvec$ from \eqref{Su_c_opt} results in the tightest WMSE lower bound that can be obtained from \eqref{Su_substitution_CS_newer}.\\
\indent
From the equality condition of Cauchy–-Schwarz, equality in \eqref{Su_substitution_CS} is obtained {\em iff}
\be \label{Su_equality_cond_preview_W}
\Wmat^{\frac{1}{2}}(\hat{\thetavec}-\thetavec_0)=\zeta_{\thetavecsmall_0,\Wmat}\Wmat^{\frac{1}{2}}\xivec_{\Wmat}(\xvec,\thetavec_0),
\ee
where $\zeta_{\thetavecsmall_0,\Wmat}$ is a scalar that may be dependent of $\thetavec_0$ and $\Wmat$. Computing the expected squared norm of each side of \eqref{Su_equality_cond_preview_W}, substituting \eqref{Su_DENOM_CS_next} and \eqref{Su_c_opt}, and using \eqref{WMSE} and \eqref{pinv_A_dagger}, we obtain
\be\label{Su_equality_cond_preview_W11}
\begin{split}
{\text{WMSE}}_{\hat{\thetavecsmall}}(\thetavec_0)&=\zeta_{\thetavecsmall_0,\Wmat}^2\psivec_{\Wmat}^T(\thetavec_0)\Gammamat_{\Umat,\Wmat}^{\dagger}(\thetavec_0)\psivec_{\Wmat}(\thetavec_0)\\
&=\zeta_{\thetavecsmall_0,\Wmat}^2B_{\text{LU-CCRB}}(\thetavec_0,\Wmat),
\end{split}
\ee
where the last equality is obtained by substituting \eqref{evec_define} and using \eqref{total_CR_bound_CS}. Thus, for obtaining equality in \eqref{total_CR_bound_CS_pre_W}, we require $\zeta_{\thetavecsmall_0,\Wmat}=\pm 1$. It can be verified that in order for $\hat{\thetavec}$ from \eqref{Su_equality_cond_preview_W} to satisfy \eqref{sec_localU_final}, we must require 
\be\label{Su_zeta_solution}
\zeta_{\thetavecsmall_0,\Wmat}=1.
\ee
Finally, by substituting \eqref{Su_b_aux_W} and \eqref{Su_zeta_solution} in \eqref{Su_equality_cond_preview_W} and then, substituting \eqref{Su_cvec_define_W}, \eqref{evec_define}, and \eqref{Su_c_opt}, we obtain \eqref{equality_cond_Prop_W}.

\section{Proof of Proposition \ref{PROP_order}}\label{App_PROP_order}
Consider the estimator
\begin{multline}\label{CCRB_efficient}
\hat{\thetavec}_{\text{CCRB}}\\
=\thetavec_0+\Umat(\thetavec_0)\left(\Umat^T(\thetavec_0)\Jmat(\thetavec_0)\Umat(\thetavec_0)\right)^{\dagger}\Umat^T(\thetavec_0)\upsilonvec(\xvec,\thetavec_0).
\end{multline}	
From the definition of $\hat{\thetavec}_{\text{CCRB}}$ in \eqref{CCRB_efficient} and by using \eqref{smoothness}, it can be verified that
\be \label{Su_point_wise_chi_unb_proof}
\bvec_{\hat{\thetavecsmall}_{\text{CCRB}}}(\thetavec_0)=\zerovec,
\ee
{\it i.e.} $\hat{\thetavec}_{\text{CCRB}}$ satisfies the first condition for local $\mathcal{X}$-unbiasedness from \eqref{point_wise_chi_unb}.
In order to show that $\hat{\thetavec}_{\text{CCRB}}$ satisfies the second condition for local $\mathcal{X}$-unbiasedness from \eqref{local_wise_chi_unb} under the assumption in \eqref{yonina_assume}, we prove the following Lemmas.
\begin{lemma} \label{Su_lamma_app1}
Eq. \eqref{yonina_assume} is satisfied {\em iff} 
\be\label{Su_yonina_assume_new}
\mathcal{R}\left(\Umat^T(\thetavec_0)\right)\subseteq\mathcal{R}\left(\Umat^T(\thetavec_0)\Jmat(\thetavec_0)\Umat(\thetavec_0)\right).
\ee
\end{lemma}
\begin{proof}
Assume that \eqref{yonina_assume} is satisfied and let $\avec_1\in\mathcal{R}\left(\Umat^T(\thetavec_0)\right)$. Thus, there exists a vector $\bvec_1$ s.t. $\avec_1=\Umat^T(\thetavec_0)\bvec_1$. Let $\tilde{\avec}_1\define\Umat(\thetavec_0)\avec_1=\Umat(\thetavec_0)\Umat^T(\thetavec_0)\bvec_1$. Thus, $\tilde{\avec}_1\in\mathcal{R}\left(\Umat(\thetavec_0)\Umat^T(\thetavec_0)\right)$ and from \eqref{yonina_assume}, $\tilde{\avec}_1\in\mathcal{R}\left(\Umat(\thetavec_0)\Umat^T(\thetavec_0)\Jmat(\thetavec_0)\Umat(\thetavec_0)\Umat^T(\thetavec_0)\right)$. Therefore, there exists $\tilde{\bvec}_1$ s.t.
\be\label{Su_Eq4_lemma1_s1}
\tilde{\avec}_1=\Umat(\thetavec_0)\Umat^T(\thetavec_0)\Jmat(\thetavec_0)\Umat(\thetavec_0)\Umat^T(\thetavec_0)\tilde{\bvec}_1
\ee
and consequently,
\be\label{Su_Eq5_lemma1_s1}
\avec_1=\Umat^T(\thetavec_0)\Jmat(\thetavec_0)\Umat(\thetavec_0)\left(\Umat^T(\thetavec_0)\tilde{\bvec}_1\right).
\ee
From \eqref{Su_Eq5_lemma1_s1} it can be seen that $\avec_1\in\mathcal{R}\left(\Umat^T(\thetavec_0)\Jmat(\thetavec_0)\Umat(\thetavec_0)\right)$ and \eqref{Su_yonina_assume_new} is obtained.\\
\indent
Next, we assume that \eqref{Su_yonina_assume_new} is satisfied and let $\avec_2\in\mathcal{R}\left(\Umat(\thetavec_0)\Umat^T(\thetavec_0)\right)$. Thus, there exists a vector $\bvec_2$ s.t. 
\be\label{Su_Eq1_lemma1_s2}
\avec_2=\Umat(\thetavec_0)\Umat^T(\thetavec_0)\bvec_2.
\ee
Left multiplying \eqref{Su_Eq1_lemma1_s2} by $\Umat^T(\thetavec_0)$ and substituting \eqref{two}, one obtains
\be\label{Su_Eq2_lemma1_s2}
\Umat^T(\thetavec_0)\avec_2=\Umat^T(\thetavec_0)\bvec_2.
\ee
Let $\tilde{\bvec}_2\define\Umat^T(\thetavec_0)\bvec_2$, $\tilde{\bvec}_2\in\mathcal{R}\left(\Umat^T(\thetavec_0)\right)$. Then, from \eqref{Su_Eq1_lemma1_s2} 
\be\label{Su_Eq4_new_lemma1_s2}
\Umat(\thetavec_0)\tilde{\bvec}_2=\avec_2.
\ee
From \eqref{Su_yonina_assume_new}, $\tilde{\bvec}_2\in\mathcal{R}\left(\Umat^T(\thetavec_0)\Jmat(\thetavec_0)\Umat(\thetavec_0)\right)$. Thus, there exists a vector $\tilde{\avec}_2$ s.t.
\be\label{Su_Eq4_lemma1_s2}
\tilde{\bvec}_2=\Umat^T(\thetavec_0)\Jmat(\thetavec_0)\Umat(\thetavec_0)\tilde{\avec}_2.
\ee
Left multiplying \eqref{Su_Eq4_lemma1_s2} by $\Umat(\thetavec_0)$, substituting \eqref{Su_Eq4_new_lemma1_s2}, and using \eqref{two}, one obtains
\be\label{Su_Eq6_lemma1_s2}
\avec_2=\Umat(\thetavec_0)\Umat^T(\thetavec_0)\Jmat(\thetavec_0)\Umat(\thetavec_0)\Umat^T(\thetavec_0)\left(\Umat(\thetavec_0)\tilde{\avec}_2\right).
\ee
From \eqref{Su_Eq6_lemma1_s2} it can be seen that $\avec_2\in\mathcal{R}\left(\Umat(\thetavec_0)\Umat^T(\thetavec_0)\Jmat(\thetavec_0)\Umat(\thetavec_0)\Umat^T(\thetavec_0)\right)$ and \eqref{yonina_assume} is obtained.
\end{proof}
\begin{lemma}\label{Su_lamma_app2}
Assume that \eqref{yonina_assume} holds. Then,
\be\label{Su_U_equal_appendix}
\begin{split}
&\Umat(\thetavec_0)\\
&=\Umat(\thetavec_0)\left(\Umat^T(\thetavec_0)\Jmat(\thetavec_0)\Umat(\thetavec_0)\right)^{\dagger}\left(\Umat^T(\thetavec_0)\Jmat(\thetavec_0)\Umat(\thetavec_0)\right).
\end{split}
\ee
\end{lemma}
\begin{proof}
In Lemma \ref{Su_lamma_app1}, we showed that \eqref{yonina_assume} is satisfied {\em iff} \eqref{Su_yonina_assume_new} is satisfied. Let $\deltavec_m(\thetavec_0)$ denote the $m$th column of $\Umat^{T}(\thetavec_0)$, $\forall m=1,\ldots,M$. Then, $\deltavec_m(\thetavec_0)\in\mathcal{R}\left(\Umat^T(\thetavec_0)\right)$ and under the assumption in \eqref{Su_yonina_assume_new},
$\deltavec_m(\thetavec_0)\in\mathcal{R}\left(\Umat^T(\thetavec_0)\Jmat(\thetavec_0)\Umat(\thetavec_0)\right)$. Thus,
\be\label{Su_m_upsilon_equality}
\begin{split}
&\deltavec_m(\thetavec_0)\\
&=\Pmat_{\Umat^T\Jmat\Umat}(\thetavec_0)\deltavec_m(\thetavec_0)\\
&=\left(\Umat^T(\thetavec_0)\Jmat(\thetavec_0)\Umat(\thetavec_0)\right)\left(\Umat^T(\thetavec_0)\Jmat(\thetavec_0)\Umat(\thetavec_0)\right)^{\dagger}\deltavec_m(\thetavec_0),
\end{split}
\ee
where the last equality stems from the definition of orthogonal projection matrix and the r.h.s. of \eqref{Su_m_upsilon_equality} is the $m$th column of the matrix $\left(\Umat^T(\thetavec_0)\Jmat(\thetavec_0)\Umat(\thetavec_0)\right)\left(\Umat^T(\thetavec_0)\Jmat(\thetavec_0)\Umat(\thetavec_0)\right)^{\dagger}\Umat^T(\thetavec_0)$, $\forall m=1,\ldots,M$. Therefore, we obtain the equality
\be\label{Su_U_upsilon_equality}
\begin{split}
&\Umat^T(\thetavec_0)\\
&=\left(\Umat^T(\thetavec_0)\Jmat(\thetavec_0)\Umat(\thetavec_0)\right)\left(\Umat^T(\thetavec_0)\Jmat(\thetavec_0)\Umat(\thetavec_0)\right)^{\dagger}\Umat^T(\thetavec_0).
\end{split}
\ee
Taking the transpose of \eqref{Su_U_upsilon_equality}, we obtain \eqref{Su_U_equal_appendix}.
\end{proof}
The bias gradient of $\hat{\thetavec}_{\text{CCRB}}$ at $\thetavec_0$ multiplied by $\Umat(\thetavec_0)$ is given by
\be\label{Su_bias_gradient_CCRB}
\begin{split}
&\nabla_{\thetavecsmall}\left.\bvec_{\hat{\thetavecsmall}_{\text{CCRB}}}(\thetavec)\right|_{\thetavecsmall_0}\Umat(\thetavec_0)\\
&=-\Umat(\thetavec_0)+\Umat(\thetavec_0)\left(\Umat^T(\thetavec_0)\Jmat(\thetavec_0)\Umat(\thetavec_0)\right)^{\dagger}\\
&~~~\times\left(\Umat^T(\thetavec_0)\Jmat(\thetavec_0)\Umat(\thetavec_0)\right).
\end{split}
\ee
From Lemma \ref{Su_lamma_app2} under the assumption in \eqref{yonina_assume}, $\hat{\thetavec}_{\text{CCRB}}$ satisfies \eqref{local_wise_chi_unb}. Thus, $\hat{\thetavec}_{\text{CCRB}}$ from \eqref{CCRB_efficient} is a locally $\mathcal{X}$-unbiased estimator in the vicinity of $\thetavec_0$. It is mentioned in Subsection \ref{subsec:Local C-unbiasedness} that a locally $\mathcal{X}$-unbiased estimator is also a locally C-unbiased estimator, in the vicinity of $\thetavec_0$, for any positive semidefinite matrix $\Wmat$. Thus, $\hat{\thetavec}_{\text{CCRB}}$ is a locally C-unbiased estimator in the vicinity of $\thetavec_0\in\Theta_\fvec$ and consequently, from Theorem \ref{T3}, the LU-CCRB is a lower bound on the WMSE of $\hat{\thetavec}_{\text{CCRB}}$, {\it i.e.}
\be \label{Su_total_CR_bound_CS_pre_order_proof}
{\rm{E}}\left[(\hat{\thetavec}_{\text{CCRB}}-\thetavec_0)^T\Wmat(\hat{\thetavec}_{\text{CCRB}}-\thetavec_0)\right]\geq B_{\text{LU-CCRB}}(\thetavec_0,\Wmat).
\ee
By computing the WMSE of the estimator from \eqref{CCRB_efficient}, one obtains
\be \label{Su_total_CR_bound_CS_pre_order}
\begin{split}
&{\rm{E}}\left[(\hat{\thetavec}_{\text{CCRB}}-\thetavec_0)^T\Wmat(\hat{\thetavec}_{\text{CCRB}}-\thetavec_0)\right]\\
&={\rm{E}}\left[\upsilonvec^T(\xvec,\thetavec_0)\Umat(\thetavec_0)\left(\Umat^T(\thetavec_0)\Jmat(\thetavec_0)\Umat(\thetavec_0)\right)^{\dagger}\Umat^T(\thetavec_0)\right.\\
&~~~\left.\times\Wmat\Umat(\thetavec_0)\left(\Umat^T(\thetavec_0)\Jmat(\thetavec_0)\Umat(\thetavec_0)\right)^{\dagger}\Umat^T(\thetavec_0)\upsilonvec(\xvec,\thetavec_0)\right].
\end{split}
\ee
By using the linearity of the trace and expectation operators and the trace operator's properties, we can rewrite \eqref{Su_total_CR_bound_CS_pre_order} as
\be \label{Su_W_CCRB_next}
\begin{split}
&{\rm{E}}\left[(\hat{\thetavec}_{\text{CCRB}}-\thetavec_0)^T\Wmat(\hat{\thetavec}_{\text{CCRB}}-\thetavec_0)\right]\\
&={\rm{Tr}}\left(\Umat^T(\thetavec_0){\rm{E}}\left[\upsilonvec(\xvec,\thetavec_0)\upsilonvec^T(\xvec,\thetavec_0)\right]\Umat(\thetavec_0)\right.\\
&~~~\left.\times\left(\Umat^T(\thetavec_0)\Jmat(\thetavec_0)\Umat(\thetavec_0)\right)^{\dagger}\left(\Umat^T(\thetavec_0)\Wmat\Umat(\thetavec_0)\right)\right.\\
&~~~\left.\times\left(\Umat^T(\thetavec_0)\Jmat(\thetavec_0)\Umat(\thetavec_0)\right)^{\dagger}\right)\\
&={\rm{Tr}}\left(\left(\Umat^T(\thetavec_0)\Jmat(\thetavec_0)\Umat(\thetavec_0)\right)^{\dagger}\left(\Umat^T(\thetavec_0)\Wmat\Umat(\thetavec_0)\right)\right)\\
&=B_{\text{CCRB}}(\thetavec_0,\Wmat),
\end{split}
\ee
where the second equality is obtained by substituting \eqref{FIM} and using trace properties and \eqref{pinv_A_dagger}. The last equality stems from \eqref{W_CCRB_W}. By substituting \eqref{Su_W_CCRB_next} in the l.h.s. of \eqref{Su_total_CR_bound_CS_pre_order_proof}, we obtain \eqref{CCRB_is_higher_W}.

\section{Proof of Proposition \ref{PROP_C_unb_Sphere}}\label{App_PROP_C_unb_Sphere}
Under the model in \eqref{model1} and \eqref{H_equation}, $\Hmat^T\xvec\sim N(\beta\thetavec,\beta\sigma^2\Imat_M)$. Consequently, it can be verified from \cite{PUKKILA} that the conditional pdf of the random vector $\frac{\hat{\thetavec}_{\text{CML}}}{\rho}$ given the random variable $\nu\define{\frac{\|\Hmat^T\xvec\|}{\beta^{\frac{1}{2}}\sigma}}$ is von Mises-Fisher (see e.g. \cite{PUKKILA,DIRECTIONAL}) with mean direction $\frac{\thetavecsmall}{\|\thetavecsmall\|}$ and concentration parameter $\frac{\nu \beta^{\frac{1}{2}}\|\thetavecsmall\|}{\sigma}$. Therefore, by using the properties of von Mises-Fisher distribution, as appear in \cite[pp. 168-169]{DIRECTIONAL}, it can be verified that the conditional expectation of $\frac{\hat{\thetavec}_{\text{CML}}}{\rho}$ given $\nu$ is
\be\label{CML_Cond_expect1}
\frac{1}{\rho}{\rm{E}}\left[\hat{\thetavec}_{\text{CML}}|\nu;\thetavec\right]=\frac{I_{\frac{M}{2}}\left(\frac{\nu \beta^{\frac{1}{2}}\|\thetavecsmall\|}{\sigma}\right)}{I_{\frac{M}{2}-1}\left(\frac{\nu \beta^{\frac{1}{2}}\|\thetavecsmall\|}{\sigma}\right)}\frac{\thetavec}{\|\thetavec\|},~\forall\thetavec\in\mathbb{R}^M,
\ee
where $I_m$ is the modified Bessel function of order $m$. By left multiplying \eqref{CML_Cond_expect1} by $\rho\Umat^T(\thetavec)$ and substituting \eqref{U_sphere_equality}, one obtains 
\be\label{CML_Cond_expect1_next}
\Umat^T(\thetavec){\rm{E}}\left[\hat{\thetavec}_{\text{CML}}|\nu;\thetavec\right]=\zerovec,~\forall\thetavec\in\mathbb{R}^M.
\ee
Consequently, for $\Wmat=\Imat_M$
\be \label{unbiased_cond_nec_exam1}
\begin{split}
\Umat^T(\thetavec)\Wmat\bvec_{\hat{\thetavecsmall}_{\text{CML}}}(\thetavec)&=\Umat^T(\thetavec){\rm{E}}[\hat{\thetavec}_{\text{CML}};\thetavec]\\
&={\rm{E}}\left[\Umat^T(\thetavec){\rm{E}}[\hat{\thetavec}_{\text{CML}}|\nu;\thetavec];\thetavec\right]=\zerovec,
\end{split}
\ee
$\forall \thetavec\in \Theta_\fvec$, where the first equality is obtained by substituting \eqref{bias_definition} and \eqref{U_sphere_equality}. The second equality is obtained by using the law of total expectation \cite{PAPOULIS} and the last equality is obtained by substituting \eqref{CML_Cond_expect1_next}. Thus, for $\Wmat=\Imat_M$, the CML estimator from \eqref{CML_equation1} satisfies \eqref{unbiased_cond_nec}, {\it i.e.} it is uniformly C-unbiased. 

\section{Derivation of \eqref{LU_CCRB_examp2}}\label{App_LU_CCRB_examp2}
By using \eqref{U_freq_define} and \eqref{W_I2}, we obtain
\be\label{evec_define_examp2}
{\rm{vec}}\left(\Umat^T(\thetavec)\Wmat\Umat(\thetavec)\right)=[1,0,0,0]^T,~\forall\thetavec\in\Theta_\fvec.
\ee
By taking the gradient of each column of $\Umat(\thetavec)$ from \eqref{U_freq_define} and using \eqref{amplitude_constraint}, we obtain
\be\label{V12_examp2}
\Vmat_1(\thetavec)=\frac{1}{c^3}\begin{bmatrix}
-\theta_1\theta_2 & \theta_1^2 & 0\\
-\theta_2^2 & \theta_1\theta_2 & 0\\
0 & 0 & 0
\end{bmatrix},~\Vmat_2(\thetavec)=\zerovec,~\forall\thetavec\in\Theta_\fvec.
\ee
By substituting \eqref{U_freq_define}, \eqref{W_I2}, and \eqref{V12_examp2} in \eqref{S_define} and using \eqref{amplitude_constraint}, we obtain
\be\label{S_define12_examp2}
\Smat^{(1)}_{\Wmat}(\thetavec)=\frac{1}{c^2}\begin{bmatrix}
-\theta_1 & 0\\
-\theta_2 & 0\\
0 & 0
\end{bmatrix},~\Smat^{(2)}_{\Wmat}(\thetavec)=\zerovec,~\forall\thetavec\in\Theta_\fvec.
\ee
Then, by substituting \eqref{W_I2} and \eqref{S_define12_examp2} in \eqref{Gamma_block_define} and using \eqref{amplitude_constraint} and the block structure of $\Cmat_{\Umat,\Wmat}(\thetavec)$, one obtains
\be\label{Gamma_block_define_examp2}
\Cmat_{\Umat,\Wmat}(\thetavec)=\begin{bmatrix}
\frac{1}{c^2} & 0 & 0 & 0\\
0 & 0 & 0 & 0\\
0 & 0 & 0 & 0\\
0 & 0 & 0 & 0\\
\end{bmatrix},~\forall\thetavec\in\Theta_\fvec.
\ee
By substituting \eqref{Gamma_block_define_examp2}, \eqref{U_freq_define}, \eqref{W_I2}, and \eqref{FIM_examp2} in \eqref{Gammamat_define} and using \eqref{amplitude_constraint}, we obtain
\be \label{Gammamat_define_examp2}
\Gammamat_{\Umat,\Wmat}(\thetavec)=\begin{bmatrix}
\frac{1}{c^2}+\frac{2L}{\sigma^2} & -\frac{cL(2l_1+L-1)}{\sigma^2} & 0 & 0\\
-\frac{cL(2l_1+L-1)}{\sigma^2} & \frac{2c^2 \sum_{l=l_1}^{l_1+L-1}l^2}{\sigma^2} & 0 & 0\\
0 & 0 & 0 & 0\\
0 & 0 & 0 & 0\\
\end{bmatrix},
\ee
$\forall\thetavec\in\Theta_\fvec$. By substituting \eqref{evec_define_examp2} and \eqref{Gammamat_define_examp2} in \eqref{total_CR_bound_CS}, one obtains
\be \label{LU_CCRB_examp2_supp1}
\begin{split}
&B_{\text{LU-CCRB}}(\thetavec,\Wmat)=\left(\frac{1}{c^2}+\frac{2L}{\sigma^2}-\frac{1}{\sigma^2}\frac{L^2(2l_1+L-1)^2}{2\sum_{l=l_1}^{l_1+L-1}l^2}\right)^{-1}\\
&=\left(\frac{1}{c^2}+\frac{2L}{\sigma^2}-\frac{1}{\sigma^2}\frac{L^2(2l_1+L-1)^2}{2Ll_1(l_1+L-1)+\frac{L(2L-1)(L-1)}{3}}\right)^{-1},
\end{split}
\ee
$\forall\thetavec\in\Theta_\fvec$, where the second equality is obtained by substituting $\sum_{l=l_1}^{l_1+L-1}l^2=Ll_1(l_1+L-1)+\frac{L(2L-1)(L-1)}{6}$. By applying simple algebraic manipulations on \eqref{LU_CCRB_examp2_supp1}, we obtain
\be \label{LU_CCRB_examp2_supp_major}
\begin{split}
&B_{\text{LU-CCRB}}(\thetavec,\Wmat)\\
&=\left(\frac{1}{c^2}+\frac{1}{\sigma^2}\frac{L(L-1)(L+1)}{6l_1^2+6(L-1)l_1+(2L-1)(L-1)}\right)^{-1}.
\end{split}
\ee
Finally, \eqref{LU_CCRB_examp2} is obtained by substituting \eqref{CCRB_examp2} into \eqref{LU_CCRB_examp2_supp_major}.

\section{Alternative derivation of LU-CCRB}\label{App_Alternative derivation of LU-CCRB}
In this section, we derive the LU-CCRB from \eqref{total_CR_bound_CS} under Conditions \ref{cond1CRB}-\ref{cond3CRB}, by solving a constrained minimization problem. For simplicity of derivation, we assume that $\Wmat=\Imat_M$ and that $\Gammamat_{\Umat,\Imat_M}(\thetavec_0)$ is a nonsingular matrix. The alternative derivation of LU-CCRB for a general weighting matrix, $\Wmat$, can be obtained in a similar manner. We minimize the MSE matrix trace at $\thetavec_0\in\Theta_\fvec$ w.r.t. the estimator $\hat{\thetavec}$ under the constraint that $\hat{\thetavec}$ is locally C-unbiased in the vicinity of $\thetavec_0$ for $\Wmat=\Imat_M$, {\it i.e.} satisfies \eqref{point_wise_unb_I} and \eqref{sec_localU_final_I}. Thus, the following constrained minimization problem is defined
\be \label{Su_CR_min3}
\begin{aligned}
&\underset{\hat{\thetavecsmall}}{\min}\left\{{\rm{E}}\left[(\hat{\thetavec}-\thetavec_0)^T(\hat{\thetavec}-\thetavec_0)\right]\right\}~~\text{s.t.}\\
&1)~\Umat^T(\thetavec_0)\bvec_{\hat{\thetavecsmall}}(\thetavec_0)=\zerovec,\\
&2)~\bvec_{\hat{\thetavecsmall}}^T(\thetavec_0)\Vmat_m(\thetavec_0)\Umat(\thetavec_0)+\uvec_m^T(\thetavec_0)\Dmat_{\hat{\thetavecsmall}}(\thetavec_0)\Umat(\thetavec_0)=\zerovec,\\
&~~~\forall m=1,\ldots,M-K.
\end{aligned}
\ee
Then, by substituting \eqref{bias_definition} and \eqref{log_expect} in \eqref{Su_CR_min3} and applying simple algebraic manipulations, the Lagrangian of the optimization problem from \eqref{Su_CR_min3} can be written as
\be\label{Su_Lagrange_CR}
\begin{split}
&L(\hat\thetavec,\{\lambdavec_k\}_{k=0}^{M-K})\\
&={\rm{E}}\bigg[(\hat{\thetavec}-\thetavec_0)^T(\hat{\thetavec}-\thetavec_0)-2\lambdavec_0^T\Umat^T(\thetavec_0)(\hat{\thetavec}-\thetavec_0)\\
&~~~-2\sum_{k=1}^{M-K}\lambdavec_k^T\bigg(\Umat^T(\thetavec_0)\Vmat_k^T(\thetavec_0)(\hat{\thetavec}-\thetavec_0)-\Umat^T(\thetavec_0)\uvec_k(\thetavec_0)\\
&~~~+\Umat^T(\thetavec_0)\upsilonvec(\xvec,\thetavec_0)\uvec_k^T(\thetavec_0)(\hat{\thetavec}-\thetavec_0)\bigg)\bigg],
\end{split}
\ee
where $\lambdavec_k\in\mathbb{R}^{M-K},~k=0,1,\ldots,M-K$, are the Lagrange multipliers. By completing the square of \eqref{Su_Lagrange_CR}, one obtains
\be \label{Su_Lagrange_CR_complete}
\begin{split}
&L(\hat\thetavec,\{\lambdavec_k\}_{k=0}^{M-K})={\rm{E}}\bigg[\left(\hat{\thetavec}-\thetavec_0-\dvec_1(\xvec,\thetavec_0)\right)^T\\
&\times\left(\hat{\thetavec}-\thetavec_0-\dvec_1(\xvec,\thetavec_0)\right)-\dvec_1^T(\xvec,\thetavec_0)\dvec_1(\xvec,\thetavec_0)+d_2(\thetavec_0)\bigg],
\end{split}
\ee
where\\ $\dvec_1(\xvec,\thetavec_0)\define\Umat(\thetavec_0)\lambdavec_0+\sum_{k=1}^{M-K}\Vmat_k(\thetavec_0)\Umat(\thetavec_0)\lambdavec_k+\sum_{k=1}^{M-K}\uvec_k(\thetavec_0)\upsilonvec^T(\xvec,\thetavec_0)\Umat(\thetavec_0)\lambdavec_k$
and\\
$d_2(\thetavec_0)\define 2\sum_{k=1}^{M-K}\lambdavec_k^T\Umat^T(\thetavec_0)\uvec_k(\thetavec_0)$. The minimization of \eqref{Su_Lagrange_CR_complete} w.r.t. $\hat{\thetavec}$ yields
\be \label{Su_Lagrange_CR_after_derivW}
\begin{split}
\hat{\thetavec}_{\text{opt}}-\thetavec_0&=\Umat(\thetavec_0)\lambdavec_0+\sum_{k=1}^{M-K}\Vmat_k(\thetavec_0)\Umat(\thetavec_0)\lambdavec_k\\
&~~~+\sum_{k=1}^{M-K}\uvec_k(\thetavec_0)\upsilonvec^T(\xvec,\thetavec_0)\Umat(\thetavec_0)\lambdavec_k.
\end{split}
\ee
Left multiplying \eqref{Su_Lagrange_CR_after_derivW} by $\Umat^{T}(\thetavec_0)$, taking expectation at $\thetavec_0$, substituting \eqref{two}, \eqref{point_wise_unb_I}, and \eqref{smoothness}, and reordering, one obtains
\begin{equation} \label{Su_finding_eta_0}
\lambdavec_0=-\sum_{k=1}^{M-K}\Umat^T(\thetavec_0)\Vmat_k(\thetavec_0)\Umat(\thetavec_0)\lambdavec_k,
\end{equation}
where \eqref{point_wise_unb_I} is the first constraint in \eqref{Su_CR_min3}. By substituting \eqref{Su_finding_eta_0} in \eqref{Su_Lagrange_CR_after_derivW} and reordering, one obtains
\be \label{Su_Lagrange_CR_after_deriv_and_eta_0_3}
\begin{split}
\hat{\thetavec}_{\text{opt}}-\thetavec_0&=\sum_{k=1}^{M-K}\bigg(\Pmat_\Umat^\bot(\thetavec_0)\Vmat_k(\thetavec_0)\Umat(\thetavec_0)\\
&~~~+\uvec_k(\thetavec_0)\upsilonvec^T(\xvec,\thetavec_0)\Umat(\thetavec_0)\bigg)\lambdavec_k.
\end{split}
\ee
The auxiliary function from \eqref{Su_b_aux_W}, which is used in the proof of Therorem \ref{T3}, is chosen based on \eqref{Su_Lagrange_CR_after_deriv_and_eta_0_3} for the case $\Wmat=\Imat_M$.\\
\indent
In order to find $\lambdavec_1,\ldots,\lambdavec_{M-K}$, we can rewrite \eqref{Su_Lagrange_CR_after_deriv_and_eta_0_3} as
\be \label{Su_Lagrange_finding_lambdavec}
\begin{split}
\hat{\thetavec}_{\text{opt}}-\thetavec&=\thetavec_0-\thetavec+\sum_{k=1}^{M-K}\Pmat_{\Umat}^\bot(\thetavec_0)\Vmat_k(\thetavec_0)\Umat(\thetavec_0)\lambdavec_k\\
&~~~+\sum_{k=1}^{M-K}\uvec_k(\thetavec_0)\lambdavec_k^T\Umat^T(\thetavec_0)\upsilonvec(\xvec,\thetavec_0).
\end{split}
\ee
Left multiplying \eqref{Su_Lagrange_finding_lambdavec} by $\uvec_m^T(\thetavec)$ and applying expectation at $\thetavec$, we obtain
\be \label{Su_Lagrange_finding_lambdavec11}
\begin{split}
&\uvec_m^T(\thetavec)\bvec_{\hat{\thetavecsmall}_{\text{opt}}}(\thetavec)\\
&=\uvec_m^T(\thetavec)(\thetavec_0-\thetavec)+\sum_{k=1}^{M-K}\uvec_m^T(\thetavec)\Pmat_{\Umat}^\bot(\thetavec_0)\Vmat_k(\thetavec_0)\Umat(\thetavec_0)\lambdavec_k\\
&~~~+\sum_{k=1}^{M-K}\uvec_m^T(\thetavec)\uvec_k(\thetavec_0)\lambdavec_k^T\Umat^T(\thetavec_0){\rm{E}}\left[\upsilonvec(\xvec,\thetavec_0);\thetavec\right],
\end{split}
\ee
$\forall m=1,\ldots,M-K$. Taking the derivative of \eqref{Su_Lagrange_finding_lambdavec11} w.r.t. $\thetavec$ at $\thetavec=\thetavec_0$, right multiplying by $\Umat(\thetavec_0)$, using \eqref{sec_localU_final_I} and \eqref{smoothness}, and applying simple algebraic manipulations, one obtains
\be \label{Su_Lagrange_finding_lambdavec_second}
\begin{split}
&\Umat^T(\thetavec_0)\uvec_m(\thetavec_0)\\
&=\sum_{k=1}^{M-K}\Umat^T(\thetavec_0)\Vmat_m^T(\thetavec_0)\Pmat_{\Umat}^\bot(\thetavec_0)\Vmat_k(\thetavec_0)\Umat(\thetavec_0)\lambdavec_k\\
&~~~+\sum_{k=1}^{M-K}\uvec_m^T(\thetavec_0)\uvec_k(\thetavec_0)\Umat^T(\thetavec_0)\Jmat(\thetavec_0)\Umat(\thetavec_0)\lambdavec_k,
\end{split}
\ee
$\forall m=1,\ldots,M-K$. Let $\lambdavec\define[\lambdavec_{1}^T,\ldots,\lambdavec_{M-K}^T]^T$. By stacking the vector equalities in \eqref{Su_Lagrange_finding_lambdavec_second}, $\forall m=1,\ldots,M-K$, and using \eqref{two}, \eqref{Gamma_block_define} with $\Wmat=\Imat_M$, Kronecker product definition, and \eqref{Gammamat_define} with $\Wmat=\Imat_M$, we obtain
\be\label{Su_Lagrange_finding_lambdavec_third}
{\rm{vec}}\left(\Imat_{M-K}\right)=\Gammamat_{\Umat,\Imat_M}(\thetavec_0)\lambdavec.
\ee
Under the assumption that $\Gammamat_{\Umat,\Imat_M}(\thetavec_0)$ is a nonsingular matrix, we can write
\be\label{Su_Lagrange_finding_lambdavec_final}
\lambdavec=\Gammamat_{\Umat,\Imat_M}^{-1}(\thetavec_0){\rm{vec}}\left(\Imat_{M-K}\right)
\ee
and consequently
\be\label{Su_Lagrange_finding_lambdavec_final_k}
\lambdavec_k=\left[\Gammamat_{\Umat,\Imat_M}^{-1}(\thetavec_0){\rm{vec}}\left(\Imat_{M-K}\right)\right]_{((k-1)(M-K)+1):(k(M-K))},
\ee
$\forall k=1,\ldots,M-K$. Substituting \eqref{Su_Lagrange_finding_lambdavec_final_k} in \eqref{Su_Lagrange_CR_after_deriv_and_eta_0_3}, one obtains \eqref{equality_cond_Prop_W} for the case $\Wmat=\Imat_M$. Thus, the minimizer of \eqref{Su_CR_min3} is the estimator from \eqref{equality_cond_Prop_W}, for the case $\Wmat=\Imat_M$, whose MSE matrix trace is the minimum of \eqref{Su_CR_min3}, given by the LU-CCRB from \eqref{total_CR_bound_CS} for $\Wmat=\Imat_M$.

\section{Derivation of \eqref{Su_c_opt}}\label{App_Derivation of Su_c_opt}
First, we prove that
\be\label{Su_Gamma_U_positive}
\Gammamat_{\Umat,\Wmat}^{\dagger}(\thetavec_0)\succeq\zerovec
\ee
and that \eqref{Su_substitution_CS_newer} can be considered only for
\be\label{Su_c_restrict}
\cvec\in{\mathcal{R}}(\Gammamat_{\Umat,\Wmat}^{\dagger}(\thetavec_0)),~\cvec\neq\zerovec.
\ee
Let
\begin{equation}\label{Su_Smat_define}
\Smat_{\Wmat}(\thetavec)\define\left[\Smat^{(1)}_{\Wmat}(\thetavec),\ldots,\Smat^{(M-K)}_{\Wmat}(\thetavec)\right]
\end{equation}
and
\begin{equation}\label{Su_Tmat_define}
\Tmat_{\Wmat}(\xvec,\thetavec)\define\left[\Tmat^{(1)}_{\Wmat}(\xvec,\thetavec),\ldots,\Tmat^{(M-K)}_{\Wmat}(\xvec,\thetavec)\right].
\end{equation}
From the definitions in \eqref{Su_Smat_define}-\eqref{Su_Tmat_define} and using \eqref{Su_firstStepAppA}, it can be verified that
\be\label{Su_Cu_matEqual_W}
\Smat_{\Wmat}^T(\thetavec_0)\Wmat\Smat_{\Wmat}(\thetavec_0)=\Cmat_{\Umat,\Wmat}(\thetavec_0)
\ee
and
\be\label{Su_Ju_matEqual_W}
\begin{split}
&{\rm{E}}\left[\Tmat_{\Wmat}^T(\xvec,\thetavec_0)\Wmat\Tmat_{\Wmat}(\xvec,\thetavec_0)\right]\\
&=\left(\Umat^T(\thetavec_0)\Wmat\Umat(\thetavec_0)\right)\otimes\left(\Umat^T(\thetavec_0)\Jmat(\thetavec_0)\Umat(\thetavec_0)\right).
\end{split}
\ee
From \eqref{Su_Cu_matEqual_W}-\eqref{Su_Ju_matEqual_W}, it can be seen that $\Cmat_{\Umat,\Wmat}(\thetavec_0)$ and $\left(\Umat^T(\thetavec_0)\Wmat\Umat(\thetavec_0)\right)\otimes\left(\Umat^T(\thetavec_0)\Jmat(\thetavec_0)\Umat(\thetavec_0)\right)$ are positive semidefinite matrices. Therefore, $\Gammamat_{\Umat,\Wmat}(\thetavec_0)$, defined in \eqref{Gammamat_define}, is also a positive semidefinite matrix. In addition, from pseudo-inverse matrix properties \cite[p. 23]{MATRIX_COOKBOOK}, symmetric matrix properties \cite[p. 31]{MATRIX_COOKBOOK}, and positive semidefinite matrix properties \cite[p. 51]{MATRIX_COOKBOOK}, \eqref{Su_Gamma_U_positive} holds as well. The vector $\cvec$ can be expressed as
\be\label{c_partition_W}
\cvec=\cvec_{\mathcal{R}}+\cvec_{\mathcal{N}},\cvec_{\mathcal{R}}\in{\mathcal{R}}(\Gammamat_{\Umat,\Wmat}^{\dagger}(\thetavec_0)),\cvec_{\mathcal{N}}\in{\mathcal{N}}((\Gammamat_{\Umat,\Wmat}^{\dagger}(\thetavec_0))^T).
\ee
As can be seen from \eqref{Su_Gamma_U_positive}, $\Gammamat_{\Umat,\Wmat}^{\dagger}(\thetavec_0)$ is a symmetric matrix and therefore, $\cvec_{\mathcal{N}}\in{\mathcal{N}}(\Gammamat_{\Umat,\Wmat}^{\dagger}(\thetavec_0))$. In case $\cvec=\cvec_{\mathcal{N}}$, then from \eqref{Su_Gamma_U_positive}, pseudo-inverse matrix properties \cite[p. 23]{MATRIX_COOKBOOK}, and symmetric matrix properties \cite[p. 31]{MATRIX_COOKBOOK}, $\cvec=\cvec_{\mathcal{N}}\in{\mathcal{N}}(\Gammamat_{\Umat,\Wmat}(\thetavec_0)$. 
Therefore, in this case,
\be\label{zero_c_Gamma_W}
\cvec^T\Gammamat_{\Umat,\Wmat}(\thetavec_0)\cvec=\cvec_{\mathcal{N}}^T\Gammamat_{\Umat,\Wmat}(\thetavec_0)\cvec_{\mathcal{N}}=0
\ee
and consequently from \eqref{Su_substitution_CS_newer}
\be\label{Su_zero_c_psi}
\psivec_{\Wmat}^T(\thetavec_0)\cvec=\psivec_{\Wmat}^T(\thetavec_0)\cvec_{\mathcal{N}}=0.
\ee
In the general case, $\cvec=\cvec_{\mathcal{R}}+\cvec_{\mathcal{N}}$, by substituting \eqref{c_partition_W} in \eqref{Su_substitution_CS_newer} and using \eqref{zero_c_Gamma_W}-\eqref{Su_zero_c_psi}, we obtain
\be\label{Su_substitution_CS_newer_app}
\begin{split}
&{\rm{E}}\left[(\hat{\thetavec}-\thetavec_0)^T\Wmat(\hat{\thetavec}-\thetavec_0)\right]\left(\cvec_{\mathcal{R}}^T\Gammamat_{\Umat,\Wmat}(\thetavec_0)\cvec_{\mathcal{R}}\right)\\
&\geq\left(\psivec_{\Wmat}^T(\thetavec_0)\cvec_{\mathcal{R}}\right)^2.
\end{split}
\ee
Thus, it suffices to consider in \eqref{Su_substitution_CS_newer} only $\cvec$ from \eqref{Su_c_restrict}. By reordering \eqref{Su_substitution_CS_newer}, we obtain
\be\label{Su_substitution_CS_new}
{\rm{E}}\left[(\hat{\thetavec}-\thetavec_0)^T\Wmat(\hat{\thetavec}-\thetavec_0)\right]\geq\frac{\left(\psivec_{\Wmat}^T(\thetavec_0)\cvec\right)^2}{\cvec^T\Gammamat_{\Umat,\Wmat}(\thetavec_0)\cvec},
\ee
for any $\cvec\in{\mathcal{R}}(\Gammamat_{\Umat,\Wmat}^{\dagger}(\thetavec_0)),~\cvec\neq\zerovec$. By using an extension of Cauchy–-Schwarz inequality \cite[Eq. (2.37)]{PECARIC}, \eqref{evec_define}, and \eqref{Su_Gamma_U_positive}, it can be verified that the tightest bound in the form of the r.h.s. of \eqref{Su_substitution_CS_new} is obtained for $\cvec$ from \eqref{Su_c_opt} and is given by the LU-CCRB from \eqref{total_CR_bound_CS}.

\bibliographystyle{IEEEtran}
\bibliography{constraint}
\end{document}